\documentclass[final,floating,letterpaper]{siamltex1213}
\usepackage{latexsym,amsmath,amssymb}
\allowdisplaybreaks[4]

\usepackage{comment}
\usepackage[compress]{cite}
\usepackage{algorithm}
\usepackage{algorithmic}

\algsetup{indent=2em}

\usepackage{graphicx}
\usepackage{epstopdf}
\usepackage{bm}
\usepackage{hyperref}

\title{Modified Poisson-Nernst-Planck model with accurate Coulomb correlation in variable media}

\author{Pei Liu
\thanks{School of Mathematical Sciences and Institute of Natural Sciences,
  Shanghai Jiao Tong University, Shanghai 200240, China ({\tt hgliupei1990@sjtu.edu.cn}). Current address: Department of Mathematics, The Penn State University, University Park, PA 16802, USA.}
\and
Xia Ji
  \thanks{Academy of Mathematics and Systems Science, Chinese Academy of Sciences, Beijing 100190, China
    ({\tt jixia@lsec.cc.ac.cn}).}
\and
Zhenli Xu
\thanks{School of Mathematical Sciences, Institute of Natural Sciences,
  and MoE Key Lab of Scientific and Engineering Computing,
  Shanghai Jiao Tong University, Shanghai 200240, China ({\tt xuzl@sjtu.edu.cn}).}
}

\begin{document}

\maketitle

\begin{abstract}
We derive a set of modified Poisson-Nernst-Planck (mPNP) equations for ion transport from the variation of the free energy functional which includes the many-body Coulomb correlation in media of variable dielectric coefficient. The correlation effects are considered through the Debye charging process in which the self energy of an ion is governed by the generalized Debye-H\"uckel equation. We develop the asymptotic expansions of the self energy taking the ion radius as the small parameter such that the multiscale model can be solved efficiently by numerical methods. It is shown that the variations of the energy functional give the self-energy-modified PNP equations which satisfy a proper weak formulation. We present numerical results from different asymptotic expansions with a semi-implicit conservative numerical method and investigate the effect of the Coulomb correlation.
\end{abstract}



\pagestyle{myheadings}
\thispagestyle{plain}
\markboth{P. Liu, X. Ji and Z. Xu}
{Modified Poisson-Nernst-Planck model}

\section{Introduction}

The understanding of the transport of charged particles in micro/nanoscale systems is of importance in many natural phenomena and has wide engineering applications \cite{FPP+:RMP:2010}. The Poisson-Nernst-Planck (PNP) theory, which is composed of the Nernst-Planck equations for the ionic flux under the concentration gradient and electrical field, and the Poisson equation for the electrical potential, is a classical continuum model to describe the ionic transport in solution and has been shown to be successful in lots of applications including nanofluidics and microfluidics \cite{Schoch:RMP:08,BTA:PRE:04}, semi-conductor devices \cite{MRS:S:90} and transmembrane ion channels \cite{Eisenberg:ACPip:2011a}.

The traditional PNP model cannot make correct prediction for systems with divalent ions or with a dielectric boundary because it treats ions as point charges without short-range size effect and ignores the long-range Coulomb correlation. The improvement of the mean-field approximation has attracted wide attention in recent decades \cite{MRA:BJ:03,CKC:BJ:03, KBA:PRE:2007, EHL:JCP:2010, LZ:BJ:2011,JL:2012:JDDE, LTZ:2012:JDDE,Gillespie:MN:14, Frydel:ACP:2016}. The steric effect between particles can be considered by introducing the excess free energy functional given by solvent entropy \cite{LZ:BJ:2011}, Lennard-Jones kernel \cite{EHL:JCP:2010}, Carnahan-Starling local density approximation \cite{Giera2013,Giera2015} or modified fundamental measure theory \cite{Wu:JCP:2002,Roth:JPC:2002}, see reviews  \cite{BKS+:ACIS:2009,Gillespie:MN:14} for discussions of this issue. For the long-ranged Coulomb correlation, Bazant {\it et al.} \cite{BSK:PRL:2011} proposed a fourth-order Poisson-Boltzmann (PB) model by introducing a parameter of correlation length, which has been further developed recently \cite{bazant2012,LE:JPCB:13,Liu2016}. Another promising routine is by introducing a self-energy correction to the potential of mean force using the generalized Debye-H\"uckel equation, which can be derived under the point-charge approximation (PCA) by the Debye closure of the BBGKY hierarchy  \cite{AM:CJU:1986}, random phase approximation in density functional theory \cite{FM:PRE:2016} or the Gaussian variational field theory \cite{podgornik1989jcp,NO:EPJE:2000,NO:EPJE:2003} for equilibrium systems. The PDEs under the PCA are unstable \cite{XuMaggs:JCP:14} for strong-coupling systems because cations and anions may collapse when the interaction is strong. To go beyond PCA, Wang \cite{Wang:PRE:2010} used the variational field theory to derive a generalized Debye-H\"uckel equation for finite-sized ions with Gaussian charge distribution, which first revealed the importance of the ionic self energy for mobile ions with finite size and obtained the Born solvation energy in variable dielectric media by performing expansion using ionic radius as small parameter, and provides the theoretical framework for accounting for these many-body effects. It should be noted that the dielectric permittivity in electrolytes could depend on the ionic concentration as well as the electrical field \cite{Bikerman:PM:1942,Grahame:CR:1947,BKS+:ACIS:2009,hatlo2012,LWZ:CMS:14,Guan:PRE:16}, thus the Born solvation energy should play important role due to strong variations in the ion concentration or the electric field. Accounting for the Born energy and ionic correlation in a continuum model is essential to explain many phenomenon such as the nonmonotonic concentration dependence of the mean activity coefficient of simple electrolytes \cite{Vincze2010} and quantitative investigation of single ion activities \cite{liu2015poisson}. The extension of these models from equilibrium to time-dependent dynamics with the Fick's law is straightforward, and results in the modified PNP equations.

In this work, we derive the system free energy through the Debye charging process \cite{DH:PZ:1923b}, which accounts for the Coulomb correlation of ions with finite size and the dielectric effect through a self energy term. The self energy is described by the generalized Debye-H\"uckel equation with variable coefficients. When finite ion size is considered, the equation is high-dimensional as the local dielectric coefficient is modified by the located ion, i.e., six dimensions for real 3 dimensional systems. Furthermore, the difference between ionic size and the Debye screening length leads to a multi-scale problem. To overcome these difficulties, asymptotic expansions of the self energy by taking the ionic size as a small parameter. The resulting modified PB/PNP equations or similar versions had been used by different authors to investigate important physical features such as the double layer structure and like-charge attraction \cite{WangRui:JCP:13,MX:JCP:14,WW:JCP:2015,LMX:CiCP:17}. With the asymptotic self-energy expression, we find the connection between the chemical potential and self energy, so that the modified PNP equation can preserve a simple form and satisfy a proper energy dissipation law thus a weak formuation. We compare the results by different kind of approximations, and demonstrate the new formulations can predict correct physics in more accurate.

On the other hand, numerical solutions of the PNP and the modified PNP equations can be very challenging, and have attracted much interest recently \cite{Flavell:JCE:14,LW:JCP:14}. The modified PNP equations subject to suitable boundary conditions make new difficulties in numerical solutions, since the correlation function may lead to the local condensation of ions and the non-monotonic ionic distributions. In Section 3, we develop a semi-implicit numerical method to solve the modified PNP equations with the mass conservation. We perform numerical example to show that the method can capture correct physical features.

\section{Coulomb correlation in variable media}
To take into account the correlation contribution in the PNP framework, we start from the free energy formulation. The modified PNP equations are derived through the energetic variational approach \cite{EHL:JCP:2010,Xu:2014:CMS} such that the energy dissipation law and a proper weak formulation will be automatically fulfilled.

Let $\Omega$ be the electrolyte region and $\Gamma= \partial \Omega$ be its boundary. Suppose there are $N$ ionic species in the electrolyte with $z_i$ being the valence of species $i$ and each ion is modeled as a hard sphere with radius $a$ and a point charge in the center. In the mean-field theory, the free energy functional is given by,
\begin{equation}
F^{\mathrm{mf}}[c_1,c_2,\cdots,c_N] =  \int_{\Omega} \left[ \frac{\varepsilon(\mathbf{r})}{2} |\nabla \phi(\mathbf{r}) |^2+\sum_{i=1}^N k_B T c_i(\mathbf{r}) [\log c_i(\mathbf{r})-1]  \right] d\mathbf{r},
\label{PB_energy}
\end{equation}
where $c_i(\mathbf{r})$ is the concentration distribution of the $i$th ion species, $k_B$ is the Boltzmann constant, $T$ is the temperature, $\varepsilon(\mathbf{r})$ is the space dependent dielectric permittivity,
and the mean electrical potential $\phi(\mathbf{r})$ satisfies the Poisson's equation,
\begin{equation}
-\nabla \cdot \varepsilon(\mathbf{r}) \nabla \phi(\mathbf{r}) = \sum_{i=1}^N z_i  c_i(\mathbf{r}),
\end{equation}
with proper boundary conditions. Since the dielectric environment is inhomogeneous, we should consider the contribution $F^{\mathrm{Born}}= \sum_{i=1}^N \int_{\Omega} c_i(\mathbf{r}) U_i^{\mathrm{Born}}(\mathbf{r}) d\mathbf{r}$ from the Born energy of each ion,
\begin{equation}
\displaystyle U_i^{\mathrm{Born}}(\mathbf{r})=\frac{z_i^2e^2}{2}\lim_{\mathbf{r}'\to\mathbf{r}} \left[G_0(\mathbf{r},\mathbf{r}')- \frac{1}{4\pi \varepsilon_0 |\mathbf{r}-\mathbf{r}'|} \right], \label{born}
\end{equation}
with $G_0$ defined below in Eq. \eqref{g0}.
Therefore the finiteness of the ionic radius $a$ is essential for variable media as the zero-radius limit leads to the divergence of the Born solvation energy. Here $\varepsilon$ is not assumed to be concentration dependent nor field dependent for simplicity, or we need extra terms to describe the polarization.

\subsection{Debye charging process}
The Debye charging process considers a hypothetical process, in which all ions are charged with the same rate from zero charge to its full charge. The work done during the process is equal to the free energy difference and thus can be used to calculate ionic chemical potential or activity coefficients \cite{DH:PZ:1923b}.
When a source charge is located at $\mathbf{r'}$, the dielectric permittivity is locally dependent on the source location due to the finite size of the ion,
\begin{equation}
\epsilon(\mathbf{r},\mathbf{r}') = \begin{cases}
\varepsilon_0,\hspace{1cm}  |\mathbf{r}-\mathbf{r}'|<a,\\
\varepsilon(\mathbf{r}),\hspace{0.7cm} \text{otherwise}.
\end{cases}
\end{equation}
We choose the reference system to be the ideal gas, i.e., its particle-particle interaction is  $v_0(\mathbf{r},\mathbf{r}')=0$, and increase the valences of all particles up to full charges: At stage $\lambda$ with $0\leq\lambda\leq1$, the valence of the $i$th species becomes $\lambda z_i$, and the pairwise potential is thus,
\begin{equation}
v_{ij}(\mathbf{r},\mathbf{r}';\lambda)= \lambda^2 v_{ij}(\mathbf{r},\mathbf{r}') =\lambda^2 z_i z_j e^2 G_0(\mathbf{r},\mathbf{r}'),
\end{equation}
where $G_0$ is governed by,
\begin{equation}
-\nabla \cdot \epsilon(\mathbf{r},\mathbf{r}') \nabla G_0(\mathbf{r},\mathbf{r}')= \delta(\mathbf{r}-\mathbf{r'}). \label{g0}
\end{equation}
When $\lambda=0$, it is the ideal gas reference system, and when $\lambda=1$, it becomes the real electrolyte system of interest.

We introduce the standard perturbation theory for many-body systems \cite{HM::2006}.
Let $F^{\mathrm{corr}}$ be the correlation free energy as a functional of single particle density. Then it can be expressed as,
\begin{equation}
F^{\mathrm{corr}}[c_1,c_2,\cdots,c_N] = \frac{1}{2} \sum_{i,j=1}^N \int_0^1 d\lambda^2 \iint_\Omega c_i(\mathbf{r}) c_j (\mathbf{r'})h_{ij}(\mathbf{r},\mathbf{r}';\lambda) v_{ij}(\mathbf{r},\mathbf{r}')  d\mathbf{r}d\mathbf{r}',
\label{charging}
\end{equation}
where $h_{ij}(\mathbf{r},\mathbf{r}';\lambda)=c_{ij}(\mathbf{r},\mathbf{r}';\lambda)/c_i(\mathbf{r})c_j(\mathbf{r'})-1$ is the total correlation function, with $c_{ij}(\mathbf{r},\mathbf{r}';\lambda)$ being the two-particle density distribution function, which can be viewed as the joint distribution of $c_i(\mathbf{r})$ and $c_j(\mathbf{r'})$. In the mean-field theory $c_{ij}(\mathbf{r},\mathbf{r}';\lambda) \approx c_i(\mathbf{r})c_j(\mathbf{r'})$ and $h_{ij}(\mathbf{r},\mathbf{r}';\lambda) \approx 0$, meaning no correlation is included. In our model, we are motivated to go beyond mean-field theory and take into account the electrostatic correlation between ions. We describe the total free energy as,
\begin{eqnarray}
F[c_1,c_2,\cdots,c_N]  & = & \int \left[ \frac{\varepsilon(\mathbf{r})}{2} |\nabla \phi(\mathbf{r}) |^2+\sum_{i=1}^N k_B T c_i(\mathbf{r}) [\log c_i(\mathbf{r})-1]  \right] d\mathbf{r} \nonumber \\
&&+  \int_\Omega d\mathbf{r} \int_0^1 d\lambda \left(2\lambda \sum_{i=1}^N c_i(\mathbf{r}) U_i(\mathbf{r};\lambda)\right).
\label{eng}
\end{eqnarray}
The second term is the correction to the classical mean field theory, where the self-energy  $U_i$  includes the Born solvation energy \eqref{born} and electrostatic correlation \eqref{charging},
\begin{equation}
U_i(\mathbf{r};\lambda) =  U_i^{\mathrm{Born}}(\mathbf{r})+ \frac{1}{2}\sum_{j=1}^N  \int_\Omega c_j (\mathbf{r'})h_{ij}(\mathbf{r},\mathbf{r}';\lambda) v_{ij}(\mathbf{r},\mathbf{r}') d\mathbf{r'}.
\label{self_a}
\end{equation}

\subsection{Self energy}

The self energy \eqref{self_a} is treated at the Debye-H\"uckel level, which means only the test ion is considered with finite radius while the surrounding mobile ions are still treated as point charges.  Let $B(\mathbf{r}')=\{\mathbf{r}: |\mathbf{r}-\mathbf{r}'|<a\}$ be the spherical domain of a test ion located at $\mathbf{r}'$, and $\chi(\mathbf{r}-\mathbf{r}')$ is the indicator function of its complementary set $B^c=\mathbb{R}^3\setminus B$, i.e., $\chi=0$ for $\mathbf{r}\in B(\mathbf{r}')$ and 1 elsewhere. Assume adding an hard sphere at location $\mathbf{r}$ will not change the surrounding ionic distribution except the interior domain of the ion sphere is inaccessible, i.e. the short-range hard core correlation is ignored. Then replace the hard sphere with  the $i$th ion, as a consequence, the electrostatic potential will change by $\phi^\delta_\lambda$, which satisfies,
\begin{equation}
-\nabla \cdot \epsilon(\mathbf{r},\mathbf{r}') \nabla \phi^\delta_\lambda(\mathbf{r};\mathbf{r'}) = \chi(\mathbf{r}-\mathbf{r}') \sum_j  \lambda z_j  e c_j(\mathbf{r}) h_{ij}(\mathbf{r},\mathbf{r'};\lambda) +\lambda z_i e \delta(\mathbf{r}-\mathbf{r'}),
\end{equation}
where the first term on the right hand side is the charge fluctuation due to  the insert of the
ion and the second term represents the point charge within the ion.
We use the Loeb closure \cite{loeb1951} in domain $B^c(\mathbf{r})$,
\begin{equation}
h_{ij}(\mathbf{r},\mathbf{r'};\lambda)\approx-\beta \lambda  z_j e \phi^\delta_\lambda(\mathbf{r};\mathbf{r'}),
\end{equation}
where $\beta=1/k_B T$ is the inverse thermal energy. Then we obtain the generalized Debye-H\"uckel (GDH) equation \cite{XML:PRE:2014,MX:JCP:14} with $\phi^\delta_\lambda=\lambda z_i e G_\lambda$,
\begin{equation}
-\nabla \cdot \epsilon(\mathbf{r},\mathbf{r}') \nabla G_\lambda(\mathbf{r};\mathbf{r'}) + 2\lambda^2 \chi(\mathbf{r}-\mathbf{r}') I(\mathbf{r}) G_\lambda(\mathbf{r};\mathbf{r'}) =\delta(\mathbf{r}-\mathbf{r'}),
\label{gDH}
\end{equation}
where  $I$ is the ionic strength defined by,
\begin{equation}
I(\mathbf{r})=\frac{1}{2}\beta e^2 \sum_{i=1}^N c_i z_i^2.
\end{equation}
Notice that the pairwise potential between the $i$th fixed test ion and the $j$th free ion is $v_{ij}(\mathbf{r},\mathbf{r}')=z_i z_j G_0(\mathbf{r},\mathbf{r}')$. The self energy of the $i$th ion defined in Eq. \eqref{self_a} is expressed as,
\begin{equation}
U_i(\mathbf{r};\lambda)=\frac{1}{2} z_i^2e^2 \lim_{\mathbf{r}'\to\mathbf{r}} \left[G_\lambda(\mathbf{r},\mathbf{r}')- \frac{1}{4\pi \varepsilon_0 |\mathbf{r}-\mathbf{r}'|}\right]\triangleq \frac{1}{2} z_i^2e^2 u(\mathbf{r};\lambda).
\label{selfenergy}
\end{equation}
Now we have the total free energy as a functional of the ionic concentrations, and the chemical potential can be computed through the functional derivatives. In literature, there is another process called the G\"untelberg charging process \cite{guntelberg1926}, where only one central ion goes through the hypothetical charging, while other ions remain fully charged. Similar derivation can also be applied to the G\"{u}ntelberg charging process, resulting in the excess chemical potential,
\begin{equation}
\mu_i^{\mathrm{ex}}(\mathbf{r}) = U_i(\mathbf{r};\lambda=1).
\label{Gun_mu}
\end{equation}
However, these two charging processes do not necessary give the same result of chemical potential \cite{Bockris1998} and unfortunately, the chemical potential given by the free energy functional \eqref{eng} does not have the simple form as \eqref{Gun_mu}.

\subsection{Asymptotic expansions}
In order to develop a PDE model that can be solved efficiently, we apply the asymptotic method for \eqref{gDH} to obtain
a simplified equation by considering the ionic radius $a$ to be a small parameter in comparison with the inverse Debye length.
We first introduce the PCA which together with the Born energy is the zeroth order asymptotics
of the GDH equation, and then discuss higher order extensions.

\subsubsection{Point charge approximation}

At the limit of $a \to 0$, ions are  point charges and the GDH equation becomes,
\begin{equation}
\displaystyle -\nabla \cdot \varepsilon(\mathbf{r})\nabla G_\lambda^\mathrm{P}(\mathbf{r},\mathbf{r'}) + 2\lambda^2 I(\mathbf{r})  G_\lambda^\mathrm{P}(\mathbf{r},\mathbf{r'}) =\delta(\mathbf{r}-\mathbf{r'}).
\label{MDH_PCA}
\end{equation}
Since the Born energy becomes singular at the zero radius, we shall consider it later on and define the self energy by,
\begin{equation}
U_i^{\mathrm{P}}(\mathbf{r};\lambda)=\frac{1}{2} z_i^2e^2 u^{\mathrm{P}}(\mathbf{r};\lambda)=\frac{z_i^2e^2}{2}\lim_{\mathbf{r}'\to\mathbf{r}} \left[ G_\lambda^\mathrm{P}(\mathbf{r};\mathbf{r'}) -\frac{1}{4\pi \varepsilon(\mathbf{r}) |\mathbf{r}-\mathbf{r}'|}\right].
\label{selfenergy_PCA}
\end{equation}
The Green's function $G_\lambda^{\mathrm{P}}(\mathbf{r},\mathbf{r}')$ can be viewed as the inverse of the modified Helmholtz operator, which is still hard to solve by numerical methods. Since we are only interested in the diagonal components, i.e., the self Green's function $G_\lambda^{\mathrm{P}}(\mathbf{r},\mathbf{r})$, the selected inversion algorithm \cite{LYM+:ATMS:2011} can be employed for the purpose \cite{XuMaggs:JCP:14}. Although the self energy is well-defined, however it ignores the excluded-volume effect and would lead to the instability for strong correlated systems due to the ionic collapse. Furthermore, since it does not include the Born energy of an ion, an appropriate treatment of it should introduce the ionic radius again.

\subsubsection{Higher-order approximations}

The Green's identities are employed to expand the asymptotic expansion into higher orders. Let us
rewrite the GDH equation \eqref{gDH} into an integral form,
\begin{eqnarray}
&&G_\lambda(\mathbf{r},\mathbf{r}')=\int \delta(\mathbf{p}-\mathbf{r}) G_\lambda(\mathbf{p},\mathbf{r}') d\mathbf{p}\nonumber\\
&=&\int \left[-\nabla_\mathbf{p} \cdot \varepsilon(\mathbf{p})  \nabla_\mathbf{p} G_\lambda^\mathrm{P}(\mathbf{p},\mathbf{r}) + 2\lambda^2 I(\mathbf{p})  G_\lambda^\mathrm{P}(\mathbf{p},\mathbf{r})\right] G_\lambda(\mathbf{p},\mathbf{r}')d\mathbf{p},
\label{green1}
\end{eqnarray}
where $\nabla_\mathbf{p}\cdot$ and $\nabla_\mathbf{p}$ are divergence and gradient operators with respect to $\mathbf{p}$. Similarly,
\begin{eqnarray}
&&G_\lambda^\mathrm{P}(\mathbf{r},\mathbf{r}')=\int  \delta(\mathbf{p}-\mathbf{r}') G_\lambda^\mathrm{P}(\mathbf{r},\mathbf{p}) d\mathbf{p}\nonumber\\
&=&\int \left[-\nabla_\mathbf{p} \cdot \epsilon(\mathbf{p},\mathbf{r}')  \nabla_\mathbf{p} G_\lambda(\mathbf{p},\mathbf{r}') + 2\lambda^2 \chi(\mathbf{p}-\mathbf{r}') I(\mathbf{p})  G_\lambda(\mathbf{p},\mathbf{r}')\right] G_\lambda^\mathrm{P}(\mathbf{r},\mathbf{p})d\mathbf{p}.~~
\label{green2}
\end{eqnarray}
We shall use the fact that $\epsilon(\mathbf{p},\mathbf{r}')=\varepsilon(\mathbf{p})$ outside $B(\mathbf{r}')$ and $\epsilon(\mathbf{p},\mathbf{r}')=\varepsilon_0$ inside $B(\mathbf{r}')$, and notice
that $\chi$ is the indicator function of $B^c$.
Subtracting \eqref{green2} from \eqref{green1} and use Green's identities, we have,
\begin{eqnarray}
&&G_\lambda(\mathbf{r},\mathbf{r}')-G_\lambda^\mathrm{P}(\mathbf{r},\mathbf{r}')\nonumber\\
&=&\int_{B^c(\mathbf{r}')}\left[ G_\lambda^\mathrm{P}(\mathbf{p},\mathbf{r})\nabla_\mathbf{p} \cdot \varepsilon(\mathbf{p})  \nabla_\mathbf{p} G_\lambda(\mathbf{p},\mathbf{r}')- G_\lambda(\mathbf{p},\mathbf{r}')\nabla_\mathbf{p} \cdot \varepsilon(\mathbf{p})  \nabla_\mathbf{p} G_\lambda^\mathrm{P}(\mathbf{p},\mathbf{r}) \right]d\mathbf{p}\nonumber\\
&&+\int_{B(\mathbf{r}')}\left[ G_\lambda^\mathrm{P}(\mathbf{p},\mathbf{r})\nabla_\mathbf{p} \cdot \varepsilon_0  \nabla_\mathbf{p} G_\lambda(\mathbf{p},\mathbf{r}')- G_\lambda(\mathbf{p},\mathbf{r}')\nabla_\mathbf{p} \cdot \varepsilon(\mathbf{p})  \nabla_\mathbf{p} G_\lambda^\mathrm{P}(\mathbf{p},\mathbf{r}) \right. \nonumber\\
&& ~~~~~~~~~~~~~~~~~~~~ \left. + 2\lambda^2 I(\mathbf{p})  G_\lambda^\mathrm{P}(\mathbf{r},\mathbf{p}) G_\lambda(\mathbf{p},\mathbf{r}') \right]d\mathbf{p} \\
&=&-\int_{\partial B(\mathbf{r}')} \varepsilon(\mathbf{p})\left[G_\lambda^\mathrm{P}(\mathbf{p},\mathbf{r})\frac{\partial}{\partial \vec{n}}G_\lambda(\mathbf{p},\mathbf{r}')-G_\lambda(\mathbf{p},\mathbf{r}')\frac{\partial}{\partial \vec{n}}G_\lambda^\mathrm{P}(\mathbf{p},\mathbf{r})\right]d\mathbf{p}\nonumber\\
&&+\int_{\partial B(\mathbf{r}')} \varepsilon_0\left[G_\lambda^\mathrm{P}(\mathbf{p},\mathbf{r})\frac{\partial}{\partial \vec{n}}G_\lambda(\mathbf{p},\mathbf{r}')-G_\lambda(\mathbf{p},\mathbf{r}')\frac{\partial}{\partial \vec{n}}G_\lambda^\mathrm{P}(\mathbf{p},\mathbf{r})\right]d\mathbf{p}\nonumber\\
&&+\int_{ B(\mathbf{r}')} G_\lambda(\mathbf{p},\mathbf{r}')\left[\nabla_\mathbf{p} \cdot (\varepsilon_0-\varepsilon(\mathbf{p}))  \nabla_\mathbf{p} G_\lambda^\mathrm{P}(\mathbf{p},\mathbf{r})
+2\lambda^2 I(\mathbf{p})  G_\lambda^\mathrm{P}(\mathbf{r},\mathbf{p})\right] d\mathbf{p},
\end{eqnarray}
where $ \partial /\partial{\vec{n}}$ is the outer normal derivatives. We use the interface conditions on $\partial B(\mathbf{r}')$, namely, functions  $G_\lambda^\mathrm{P}(\mathbf{p},\mathbf{r}')$, $G_\lambda(\mathbf{p},\mathbf{r}')$, $\epsilon(\mathbf{p},\mathbf{r}') \partial_{\vec{n}}G_\lambda(\mathbf{p},\mathbf{r}')$ and $\varepsilon(\mathbf{p})\partial_{\vec{n}} G_\lambda^\mathrm{P}(\mathbf{p},\mathbf{r}')$ are all continuous
across the interface, and we then obtain,
\begin{eqnarray}
&&G_\lambda(\mathbf{r},\mathbf{r}')-G_\lambda^\mathrm{P}(\mathbf{r},\mathbf{r}')\nonumber\\
&=&\left[1-\frac{\varepsilon_0}{\varepsilon(\mathbf{r})}\right] G_\lambda(\mathbf{r},\mathbf{r}')+ \lambda^2\int_{B(\mathbf{r}')}  \varepsilon_0  \kappa^2(\mathbf{p})  G_\lambda^\mathrm{P}(\mathbf{r},\mathbf{p}) G_\lambda(\mathbf{p},\mathbf{r}')d\mathbf{p}\nonumber\\
&&-\int_{B(\mathbf{r}')} G_\lambda(\mathbf{p},\mathbf{r}')\frac{\varepsilon_0}{\varepsilon(\mathbf{p})}\nabla_\mathbf{p} \varepsilon(\mathbf{p}) \cdot \nabla_\mathbf{p} G_\lambda^\mathrm{P}(\mathbf{p},\mathbf{r})d\mathbf{p}\nonumber\\
&&+\int_{\partial B(\mathbf{r}')} \left(\varepsilon_0-\varepsilon(\mathbf{p})\right) G_\lambda(\mathbf{p},\mathbf{r}')\frac{\partial}{\partial \vec{n}}G_\lambda^\mathrm{P}(\mathbf{p},\mathbf{r})d\mathbf{p},
\label{integral}\end{eqnarray}
where $\kappa$ is the inverse of the local Debye-length function, $\kappa(\mathbf{p})=\sqrt{2I(\mathbf{p})/\varepsilon(\mathbf{p})}$. Define the nonsingular parts of $G_\lambda$ and $G$ through,
\begin{equation}
\begin{cases}
\displaystyle G_\lambda^{\mathrm{P}}(\mathbf{r},\mathbf{r}')= \frac{1}{4\pi\sqrt{\varepsilon(\mathbf{r})\varepsilon(\mathbf{r}')}|\mathbf{r}-\mathbf{r}'|} + \Phi_\lambda^\mathrm{P}(\mathbf{r},\mathbf{r}'),\\
\displaystyle G_\lambda(\mathbf{r},\mathbf{r}')= \frac{1}{4\pi\varepsilon_0|\mathbf{r}-\mathbf{r}'|} + \Phi_\lambda(\mathbf{r},\mathbf{r}').
\end{cases}
\label{singular}
\end{equation}
Thus, in whole space of $\Omega$, $\Phi^{\mathrm{P}}_\lambda$ satisfies,
\begin{equation}
\displaystyle -\nabla \cdot \varepsilon(\mathbf{r}) \nabla \Phi_\lambda^{\mathrm{P}}(\mathbf{r},\mathbf{r}') + 2\lambda^2 I(\mathbf{r}) \Phi_\lambda^{\mathrm{P}}(\mathbf{r},\mathbf{r}') =  -\frac{\lambda^2\kappa^2(\mathbf{r})\sqrt{\varepsilon(\mathbf{r})}+\nabla^2 \sqrt{\varepsilon(\mathbf{r})}}{4\pi|\mathbf{r}-\mathbf{r}'|\sqrt{\varepsilon(\mathbf{r}')}}.
\label{non-singular}
\end{equation}
Assume $\varepsilon(\mathbf{r})$ and $\kappa(\mathbf{r})$ are smooth functions, then $\Phi_\lambda^{P}$ and $\nabla \Phi_\lambda^{P}$ are smooth too.
Since we are only interested in the diagonal entries of $G_\lambda$, by substituting \eqref{singular} into \eqref{integral}
and taking the limit $\mathbf{r}'\rightarrow \mathbf{r}$, we obtain,
\begin{equation}
\frac{\varepsilon_0}{\varepsilon(\mathbf{r})}\Phi_\lambda(\mathbf{r},\mathbf{r})-\Phi_\lambda^{\mathrm{P}}(\mathbf{r},\mathbf{r})= \lambda^2 S_1 - S_2 + S_3, \label{PHI}
\end{equation}
where $S_1$, $S_2$ and $S_3$ are the integrals given in \eqref{integral} after taking the limit $\mathbf{r}'\rightarrow \mathbf{r}$, i.e.,
\begin{equation}
\begin{cases}
\displaystyle S_1(\mathbf{r};\lambda)=\int_{B(\mathbf{r})}  \varepsilon_0 \kappa^2(\mathbf{p}) G_\lambda^\mathrm{P}(\mathbf{r},\mathbf{p}) G_\lambda(\mathbf{p},\mathbf{r})  d\mathbf{p}, \\
\displaystyle S_2(\mathbf{r};\lambda)=\int_{B(\mathbf{r})}  G_\lambda(\mathbf{p},\mathbf{r}) \frac{\varepsilon_0}{\varepsilon(\mathbf{p})} \nabla_{\mathbf{p}} \varepsilon(\mathbf{p}) \cdot \nabla_{\mathbf{p}} G_\lambda^\mathrm{P}(\mathbf{p},\mathbf{r})  d\mathbf{p},\\
\displaystyle S_3(\mathbf{r};\lambda)= \int_{\partial B(\mathbf{r})} G_\lambda(\mathbf{p},\mathbf{r}) \left(\varepsilon_0-\varepsilon(\mathbf{p})\right) \frac{\partial}{\partial \vec{n}}
G_\lambda^\mathrm{P}(\mathbf{p},\mathbf{r}) d\mathbf{p}.
\end{cases}
\end{equation}

Now we can obtain a higher-order asymptotic expansion by calculating the integrals within the ball $B$ and the integral on its boundary.
For the leading order, we can use the singular parts of $G_\lambda$ and $G_\lambda^\mathrm{P}$. This yields,
\begin{equation}
\begin{cases}
\displaystyle S_1(\mathbf{r};\lambda) \approx \int_{B(\mathbf{r})} \frac{1}{4\pi\sqrt{\varepsilon(\mathbf{r})\varepsilon(\mathbf{p})}|\mathbf{r}-\mathbf{p}|} \varepsilon_0\kappa^2(\mathbf{p}) \frac{1}{4\pi\varepsilon_0|\mathbf{r}-\mathbf{p}|}  d\mathbf{p},\\
\displaystyle S_2(\mathbf{r};\lambda) \approx \int_{B(\mathbf{r})}  \frac{\nabla_\mathbf{p}  \varepsilon(\mathbf{p})}{4\pi\varepsilon(\mathbf{p})|\mathbf{r}-\mathbf{p}|} \cdot \nabla_\mathbf{p}
\frac{1}{4\pi\sqrt{\varepsilon(\mathbf{r})\varepsilon(\mathbf{p})}|\mathbf{r}-\mathbf{p}|}  d\mathbf{p}, \\
\displaystyle S_3(\mathbf{r};\lambda) \approx -\int_{\partial B(\mathbf{r})} \frac{\varepsilon_0-\varepsilon(\mathbf{p})}{4\pi\varepsilon_0|\mathbf{r}-\mathbf{p}|}  \frac{\partial}{\partial \vec{n}} \frac{1}{4\pi\sqrt{\varepsilon(\mathbf{r})\varepsilon(\mathbf{p})}|\mathbf{r}-\mathbf{p}|} d\mathbf{p}.
\end{cases}
\end{equation}
Clearly, $S_1\sim O(a)$, $S_2 \sim O(a)$, and $S_3$ has the leading asymptotics of order $a^{-1}$,
\begin{equation}
S_3(\mathbf{r};\lambda) \sim -\left(\frac{1}{\varepsilon(\mathbf{r})}-\frac{1}{\varepsilon_0}\right) \frac{1}{4\pi a},
\end{equation}
which dominates the right hand side of Eq. \eqref{PHI}.
Since $u^\mathrm{P}=O(1)$, we can easily obtain a zeroth-order approximation of the self energy,
\begin{equation}
u^{(0)}(\mathbf{r};\lambda) =  \left(\frac{1}{\varepsilon(\mathbf{r})}-\frac{1}{\varepsilon_0}\right) \frac{1}{4\pi a}+u^{\mathrm{P}}(\mathbf{r};\lambda).
\end{equation}
This is the Born energy correction to the PCA.

In order to compute the coefficients of $O(a)$ terms, we need to evaluate the integrals up to one more order. We have,
\begin{equation}
\begin{cases}
\displaystyle S_1(\mathbf{r};\lambda) = \left(\frac{\varepsilon_0}{\varepsilon(\mathbf{r})}+1\right)\frac{\kappa^2(\mathbf{r}) a}{8\pi\varepsilon(\mathbf{r})}+O(a^2),\\
\displaystyle S_2(\mathbf{r};\lambda) = -\left(\frac{\varepsilon_0}{\varepsilon(\mathbf{r})}+1\right)\frac{a}{24\pi}\frac{\nabla^2 \varepsilon(\mathbf{r})}{\varepsilon^2(\mathbf{r})}+ O(a^2),\\
\displaystyle S_3(\mathbf{r};\lambda) = \frac{\varepsilon_0 - \varepsilon}{\varepsilon}\left(\frac{1}{4\pi \varepsilon_0 a} + \Delta u^{(1)} \right)-\frac{\varepsilon_0}{24\pi \varepsilon^3}\nabla^2 \varepsilon a+O(a^2),
\end{cases}
\end{equation}
where
\begin{equation}
\Delta u^{(1)}(\mathbf{r};\lambda)=\frac{a}{4\pi\varepsilon(\mathbf{r})}\left[\lambda^2\kappa^2(\mathbf{r})+\frac{\nabla^2 \varepsilon(\mathbf{r})}{6\varepsilon(\mathbf{r})}\right].
\label{first_add}
\end{equation}
Substituting these terms into Eq. \eqref{PHI}, we arrive at the first-order asymptotics,
\begin{equation}
u^{(1)}(\mathbf{r};\lambda)= u^{\mathrm{P}}(\mathbf{r};\lambda)+ \left[\frac{1}{\varepsilon(\mathbf{r})}-\frac{1}{\varepsilon_0}\right] \frac{1}{4\pi a} +\Delta u^{(1)}(\mathbf{r};\lambda).
\label{firstorder}
\end{equation}
In the expression, the first term is the PCA correlation energy, the second term is the Born energy due to the dielectric variation, and the two terms in $\Delta u^{(1)}$ are corrections
to the PCA due to the finite ion size and to the Born energy, respectively. When the dielectric coefficient is highly inhomogeneous, the $O(a)$ corrections become important.
In particular, the correction to the Born energy is often ignored in the continuum theory and Monte Carlo simulations of charged systems in variable media \cite{FXH:JCP:14,Guan:PRE:16},
and its effect on the ionic structure is less investigated.

The higher-order asymptotics can be obtained by further expanding the three integrals. This will lead to a tedious calculation. We only compute the correction term
from the ionic correlation and ignore the contribution from inhomogeneous dielectric media in the second-order term, which gives an additional correction,
\begin{equation}
\Delta u^{(2)}(\mathbf{r};\lambda) = \lambda^2 \kappa^2(\mathbf{r}) a^2 u^\mathrm{P}(\mathbf{r};\lambda),
\label{secondorder}
\end{equation}
to Eq. \eqref{firstorder}. We then obtain a second-order asymptotics, $u^{(2)}=u^{(1)}+\Delta u^{(2)}$, for systems with slow dielectric variations. It can be seen that the smallness of $a$ is in comparison with the local Debye length, i.e.,  $a\ll 1/\kappa$.

\subsubsection{Accuracy validation of the asymptotics}

To validate the accuracy of these asymptotic approximations, we study two special cases and present numerical calculations for a toy example with more complicated setup.

\paragraph{Case 1} In the original Debye-H\"uckel theory, the self energy is calculated in the bulk region where $\varepsilon(\mathbf{r})=\varepsilon_b$ and the ionic strength $I(\mathbf{r})=\kappa_b^2\varepsilon_b/2$ are both constants. In this case, the solution of \eqref{gDH} can be computed analytically,
\begin{equation}
G_\lambda(\mathbf{r};\mathbf{r}')=
\begin{cases}
\displaystyle \frac{e^{\lambda \kappa_b a}}{1+\lambda\kappa_b a}\frac{e^{-\lambda\kappa_b |\mathbf{r}-\mathbf{r}'|}}{4\pi \epsilon_b |\mathbf{r}-\mathbf{r}'|}, \text{ for } \mathbf{r}\in B^c(\mathbf{r}'),\\
\displaystyle \frac{1}{4\pi\varepsilon_0 |\mathbf{r}-\mathbf{r}'|} + \frac{1}{4\pi\varepsilon_b a(1+\lambda\kappa_b a)}-\frac{1}{4\pi \varepsilon_0 a}, \text{ for } \mathbf{r}\in B(\mathbf{r}').
\end{cases}
\end{equation}
In this case, the PCA Green's function is the screened Coulomb potential and thus we can easily find the PCA self energy $ u^{\mathrm{P}}(\mathbf{r};\lambda) = -\lambda \kappa_b/(4\pi\varepsilon_b).$
The exact self energy and its asymptotic expansion are,
\begin{eqnarray}
u(\mathbf{r};\lambda) &=& \frac{-\lambda \kappa_b }{4\pi\varepsilon_b (1+\lambda\kappa_b a)}+\frac{1}{4\pi a}\left(\frac{1}{\varepsilon_b}-\frac{1}{\varepsilon_0}\right)  \nonumber \\
&\sim& \frac{1}{4\pi a}\left(\frac{1}{\varepsilon_b}-\frac{1}{\varepsilon_0}\right) + u^{\mathrm{P}}(\mathbf{r};\lambda) + \frac{\lambda^2 \kappa_b^2 a}{4\pi \varepsilon_b} + \lambda^2 \kappa_b^2 a^2 u^{\mathrm{P}}(\mathbf{r};\lambda) + \cdots.
\end{eqnarray}
This analytical result is consistent with our asymptotic approximations. The comparison results are shown in Fig. \ref{test}(a).

\paragraph{Case 2} In this case, we fix $\lambda=1$ and consider the case that $\kappa(\mathbf{r})$ and $\varepsilon(\mathbf{r})$ are only functions of radial distance $r$ and satisfy the relation,
\begin{equation}
\kappa^2(r)+\frac{\nabla^2\sqrt{\varepsilon(r)}}{\sqrt{\varepsilon(r)}}=\mu^2
\end{equation}
where $\mu$ is a constant. We can solve the solution of PCA Green's function with the source at the origin $\mathbf{r}'=\mathbf{0}$,
\begin{equation}
G^{\mathrm{P}}(\mathbf{r},\mathbf{0})=\frac{ e^{-\mu r}}{4\pi r\sqrt{\varepsilon(r)\varepsilon(0)}}.
\end{equation}
The PCA self energy is $u^{\mathrm{P}}(\mathbf{0})=-\mu/(4\pi \varepsilon(0)).$
The solution of the GDH equation \eqref{gDH} is,
\begin{equation}
\displaystyle G(\mathbf{r},\mathbf{0})=\begin{cases}
\displaystyle\frac{e^{-\mu(r-a)}}{4\pi r\sqrt{\varepsilon(r)\varepsilon(a)}}\frac{1}{1+\left(\mu+\frac{\nabla \varepsilon(a)}{2\varepsilon(a)}\right)a}, \text{ if } r>a,\\
\displaystyle \frac{1}{4\pi \varepsilon_0 r}+\frac{1}{4\pi \varepsilon(a)a}\frac{1}{1+\left(\mu+\frac{\nabla \varepsilon(a)}{2\varepsilon(a)}\right)a}- \frac{1}{4\pi \varepsilon_0 a},  \text{ if } r<a,
\end{cases}
\end{equation}
which yields the self energy,
\begin{equation}
u(\mathbf{0})=\frac{1}{4\pi \varepsilon(a)a}\frac{1}{1+\left(\mu+\frac{\nabla \varepsilon(a)}{2\varepsilon(a)}\right)a}- \frac{1}{4\pi \varepsilon_0 a}.
\end{equation}
Since $\varepsilon$ is smooth and at least the second-order derivative exists, we can expand it to obtain,
$
\varepsilon(a)=\varepsilon(0)+\frac{1}{6}\nabla^2 \varepsilon a^2+\cdots.
$
We then have,
\begin{equation}
u(\mathbf{0}) \sim \frac{1}{4\pi a}\left(\frac{1}{\varepsilon(0)}-\frac{1}{\varepsilon_0}\right)+u^{\mathrm{P}}(\mathbf{0})+\frac{a}{4\pi \varepsilon} \left( \kappa^2 + \frac{\nabla^2 \varepsilon}{6\varepsilon}\right),
\end{equation}
which is consistent with the first order approximation \eqref{firstorder}.

\paragraph{Example 1} In this example, the dielectric permittivity is piecewise constant: inside the ion, $\varepsilon_0=1.1$, when $x<-1$, $\varepsilon(\mathbf{r})=0.01$, and otherwise $\varepsilon(\mathbf{r})=1$. we choose $\kappa^2_\lambda(\mathbf{r})= e^{-2r^2}+1$ to mimic the ionic distribution near a charged ion immersed in an electrolyte.  We calculate the self energy, $u(\mathbf{0}; 1)$ with different asymptotic approximations, including the 0th, the 1st and the 2nd corrections. The exact data are prepared by numerical solution of the 3D partial differential equation \eqref{gDH} by the finite-difference method \cite{IBR:CPC:1998} and the reaction potential is computed with a mesh size $h=0.03$ with $128^3$ grid points.
In order to check how accurate this mesh size is, we calculate the numerical solution of {\it Case 1} which is in agreement with the exact solution (solid line) from Fig. \ref{test}(a).
The results in Fig. \ref{test}(b) show that the 2nd correction asymptotics with a varying ionic strength and dielectric coefficient match the exact data well, and the zeroth order asymptotic solution significantly deviates from the exact data. It is also shown that the 1st order correction becomes less accurate with the increase of the ionic size, and when $a$ is large the 1st order correction is not enough to provide high accuracy.

\begin{figure}[!htbp]
\begin{center}
\includegraphics[width=0.5\textwidth]{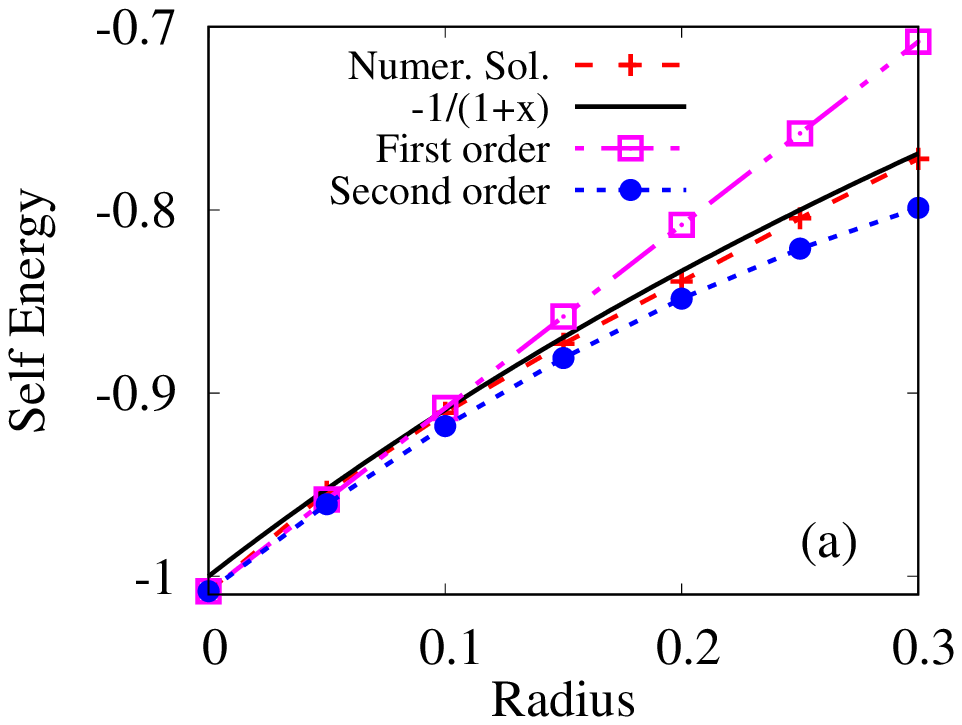}\includegraphics[width=0.5\textwidth]{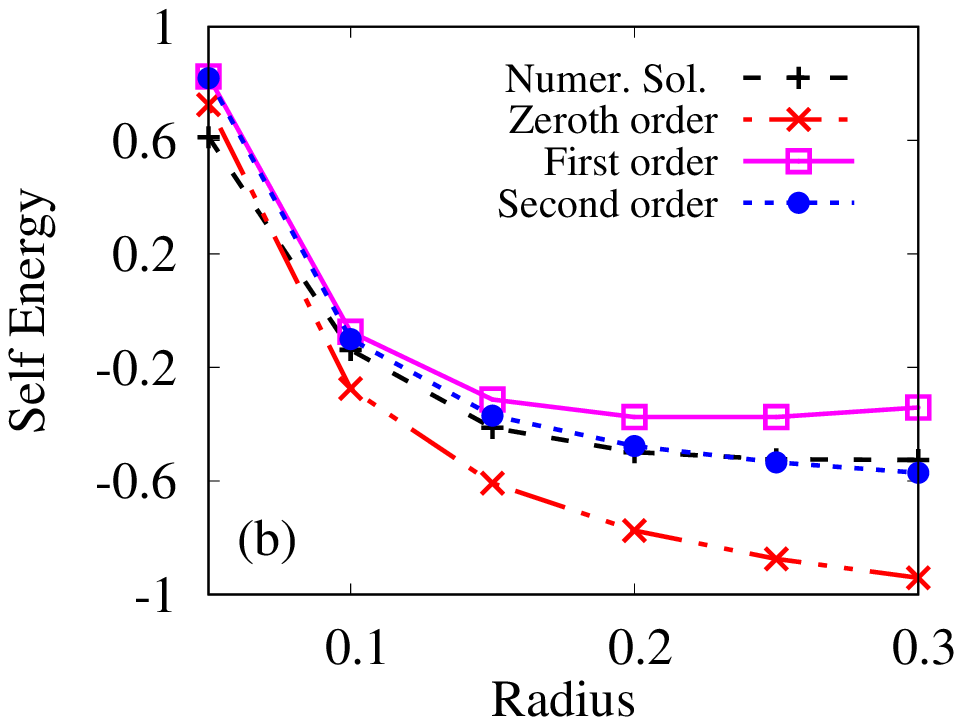}
\caption{Self energy $u(0;1)$ of an ion using different asymptotic approximations compared to the direct numerical solution. (a) Results for {\it Case 1} which has an analytical solution. (b) Results for {\it Example 1} with a varying ionic strength and dielectric coefficient.} \label{test}
\end{center}
\end{figure}

\subsection{Self-energy-modified PNP model}

The self energy defined through the point charge approximation is a functional of ionic concentration. In order to obtain the transport equations for the ions, we need to compute the functional derivative of free energy with respect to ionic concentration. First we will present some useful lemmas.

\begin{lemma} \label{lemmaCond}
The functional derivative of the PCA self energy \eqref{selfenergy_PCA} with respect to ionic concentration $c_j$ is given by,
\begin{equation}
\frac{\delta u^{\mathrm{P}}(\mathbf{r};\lambda)}{\delta c_j(\mathbf{r}')}= -\lambda^2 z_j^2 \beta e^2 G_\lambda^{\mathrm{P}}(\mathbf{r},\mathbf{r}') G_\lambda^{\mathrm{P}}(\mathbf{r}',\mathbf{r}).
\label{UdeltaC}
\end{equation}
\end{lemma}
\begin{proof}
Suppose that $\widetilde{c}_i(\mathbf{r})=c_i(\mathbf{r}) + \delta c_i(\mathbf{r})$ is the perturbed concentration, where $\delta c_i$ is small, the corresponding ionic strength is denoted as $\displaystyle \widetilde{I}(\mathbf{r})=\frac{\beta e^2}{2}\sum_i \widetilde{c}_i z_i^2 $, and the PCA self energy after perturbation is described by,
\begin{equation}
\begin{cases}
\displaystyle -\nabla \cdot \varepsilon(\mathbf{r})  \nabla \widetilde{G}_\lambda^{\mathrm{P}}(\mathbf{r},\mathbf{r}') + 2\lambda^2 \widetilde{I}(\mathbf{r})  \widetilde{G}_\lambda^{\mathrm{P}}(\mathbf{r},\mathbf{r}') =\delta(\mathbf{r}-\mathbf{r'})\\
\displaystyle \widetilde{u}(\mathbf{r};\lambda)= \lim_{\mathbf{r}' \to \mathbf{r}}\left(\widetilde{G}_\lambda^{\mathrm{P}}(\mathbf{r},\mathbf{r'})-\frac{1}{\varepsilon(\mathbf{r})}G_0(\mathbf{r},\mathbf{r'})\right).
\end{cases}
\label{VariationU}
\end{equation}
Dropping the higher order terms, the difference for the PCA Green's function equation between \eqref{MDH_PCA} and \eqref{VariationU} is,
\begin{equation}
\begin{cases}
\displaystyle -\nabla \cdot \varepsilon(\mathbf{r})  \nabla (\widetilde{G}_\lambda^{\mathrm{P}} -G_\lambda^{\mathrm{P}}) + 2\lambda^2  I(\mathbf{r})  (\widetilde{G}_\lambda^{\mathrm{P}} -G_\lambda^{\mathrm{P}}) =2\lambda^2 (I-\widetilde{I})  \widetilde{G}_\lambda^{\mathrm{P}},\\
\displaystyle \widetilde{u}^{\mathrm{P}}(\mathbf{r};\lambda)-u^{\mathrm{P}}(\mathbf{r};\lambda)= \lim_{\mathbf{r}' \to \mathbf{r}}\left(\widetilde{G}_\lambda^{\mathrm{P}}(\mathbf{r},\mathbf{r'})-G_\lambda^{\mathrm{P}}(\mathbf{r},\mathbf{r'})\right).
\end{cases}
\end{equation}
The solution can be then expressed by the Green's identity,
\begin{equation}
\displaystyle  \widetilde{G}_\lambda^{\mathrm{P}}(\mathbf{r};\mathbf{r'}) -G_\lambda^{\mathrm{P}}(\mathbf{r};\mathbf{r'}) = \int_\Omega 2\lambda^2   (I(\mathbf{p})-\widetilde{I}(\mathbf{p})) G_\lambda^{\mathrm{P}}(\mathbf{p},\mathbf{r'}) \widetilde{G}_\lambda^{\mathrm{P}}(\mathbf{r},\mathbf{p}) d\mathbf{p}.
\end{equation}
Thus the functional derivative leads to Eq. \eqref{UdeltaC}.
\end{proof}

From this lemma, we can easily see the PCA self energy satisfies the following reciprocal relation,
\begin{equation}
\frac{z_i^2 e^2}{2}\frac{\delta u_i^{\mathrm{P}}(\mathbf{r};\lambda)}{\delta c_j(\mathbf{r}')}=\frac{z_j^2 e^2}{2}\frac{\delta u_j^{\mathrm{P}}(\mathbf{r}';\lambda)}{\delta c_i(\mathbf{r})}.
\end{equation}

\begin{lemma}\label{lemmaLambda}
The derivative of the PCA self energy \eqref{selfenergy_PCA} with respect to the coupling parameter $\lambda$ is given by,
\begin{equation}
\frac{\partial u^{\mathrm{P}}(\mathbf{r};\lambda)}{\partial \lambda}= -\int_\Omega 4\lambda  I(\mathbf{p}) G_\lambda^{\mathrm{P}}(\mathbf{p},\mathbf{r}) G_\lambda^{\mathrm{P}}(\mathbf{r},\mathbf{p}) d\mathbf{p}.
\label{ulambda}
\end{equation}
\end{lemma}
\begin{proof}
Consider the system with a small perturbation in $\lambda$, $\widehat{\lambda}=\lambda + \delta \lambda$. In the new system, the generalized Debye-H\"uckel equation and the self energy become,
\begin{equation}
\begin{cases}
\displaystyle -\nabla \cdot \varepsilon(\mathbf{r})  \nabla \widehat{G_\lambda}^{\mathrm{P}}(\mathbf{r},\mathbf{r'}) + 2\widehat{\lambda}^2  I(\mathbf{r})  \widehat{G_\lambda}^{\mathrm{P}}(\mathbf{r},\mathbf{r'}) =\delta(\mathbf{r}-\mathbf{r'}),\\
\displaystyle \widehat{u}^{\mathrm{P}}(\mathbf{r};\lambda)= \lim_{\mathbf{r}' \to \mathbf{r}}\left(\widehat{G_\lambda}^{\mathrm{P}}(\mathbf{r},\mathbf{r'})-\frac{1}{\varepsilon(\mathbf{r})} G_0(\mathbf{r},\mathbf{r'})\right).
\end{cases}
\label{ChargingU2}
\end{equation}
Then the difference between Eqs. \eqref{ChargingU2} and \eqref{MDH_PCA} leads us to,
\begin{equation}
\begin{cases}
\displaystyle -\nabla \cdot \varepsilon(\mathbf{r})  \nabla (\widehat{G_\lambda}^{\mathrm{P}} -G_\lambda^{\mathrm{P}}) + 2\lambda^2  I(\mathbf{r})  (\widehat{G_\lambda}^{\mathrm{P}} -G_\lambda^{\mathrm{P}})
=2(\lambda^2-\widehat{\lambda}^2) I(\mathbf{r})   \widehat{G_\lambda}^{\mathrm{P}},\\
\displaystyle \widehat{u}^{\mathrm{P}}(\mathbf{r};\lambda)-u^{\mathrm{P}}(\mathbf{r};\lambda)= \lim_{\mathbf{r}' \to \mathbf{r}}\left(\widehat{G_\lambda}^{\mathrm{P}}(\mathbf{r},\mathbf{r'})-G_\lambda^{\mathrm{P}}(\mathbf{r},\mathbf{r'})\right).
\end{cases}
\end{equation}
We then have,
\begin{equation}
\displaystyle  \widehat{G_\lambda}^{\mathrm{P}}(\mathbf{r},\mathbf{r'}) -G_\lambda^{\mathrm{P}}(\mathbf{r},\mathbf{r'}) = \int_\Omega 2(\lambda^2-\widehat{\lambda}^2)   I(\mathbf{p}) G_\lambda^{\mathrm{P}}(\mathbf{p},\mathbf{r'}) \widehat{G_\lambda}^{\mathrm{P}}(\mathbf{r},\mathbf{p}) d\mathbf{p}.
\end{equation}
Taking the derivative with respect to $\lambda$, we get the result of expression \eqref{ulambda}.
\end{proof}

\begin{lemma} \label{lemmaPCA}
The Debye charging process of the PCA self energy gives the correlation free energy whose corresponding chemical potential is $z_i^2 e^2 u(\mathbf{r};1)/2$ for the $i$th species.
\end{lemma}

\begin{proof}
The correlation energy,
\begin{equation}
\displaystyle \mathcal{F}_{\mathrm{cor}}^{\mathrm{P}}= \int_\Omega \sum_i c_i(\mathbf{p}) z_i^2 e^2 \int_0^1 \lambda u^{\mathrm{P}}(\mathbf{p};\lambda) d\lambda d\mathbf{p}.
\label{PCA_F}
\end{equation}
The corresponding excess chemical potential is given by the functional derivative,
\begin{equation}
\displaystyle \frac{\delta }{\delta c_j(\mathbf{r})}\mathcal{F}_{\mathrm{cor}}^{\mathrm{P}}=\int_0^1 \lambda z_j^2 e^2 u^{\mathrm{P}}(\mathbf{r};\lambda) d\lambda + \int_\Omega \sum_i c_i(\mathbf{p}) z_i^2 e^2\int_0^1 \lambda \frac{\delta u^{\mathrm{P}}(\mathbf{p};\lambda)}{\delta c_j(\mathbf{r})} d\lambda d\mathbf{p}.
\label{Varfl}
\end{equation}
The second term in Eq. \eqref{Varfl} is,
\begin{eqnarray}
& &\int_\Omega \sum_i c_i(\mathbf{p}) z_i^2 e^2\int_0^1 \lambda \frac{\delta u^{\mathrm{P}}(\mathbf{p};\lambda)}{\delta c_j(\mathbf{r})} d\lambda d\mathbf{p}  \nonumber\\
& = &- \int_\Omega \sum_i c_i(\mathbf{p}) z_i^2 e^2\int_0^1  \beta e^2 z_j^2 \lambda^3  G_\lambda^{\mathrm{P}}(\mathbf{p},\mathbf{r})G_\lambda^{\mathrm{P}}(\mathbf{r},\mathbf{p}) d\lambda d\mathbf{p} \nonumber \\
& = & \int_0^1 \frac{1}{2} z_j^2e^2 \lambda^2 \frac{d u^{\mathrm{P}}(\mathbf{r};\lambda)}{d \lambda} d\lambda \nonumber\\
& = & \frac{1}{2} z_j^2 e^2 u^{\mathrm{P}}(\mathbf{r};1) -\int_0^1 \lambda z_j^2  e^2 u^{\mathrm{P}}(\mathbf{r};\lambda) d\lambda,
\end{eqnarray}
where the first two equalities use Lemmas \ref{lemmaCond} and \ref{lemmaLambda}, respectively, and the last equality is simply the result of the integration by parts.
The chemical potential can be then written as,
\begin{equation}
\displaystyle \mu_{j,\mathrm{cor}}^{\mathrm{P}}= \frac{\delta }{\delta c_j(\mathbf{r})}\mathcal{F}_{\mathrm{cor}}^{\mathrm{P}}=\frac{1}{2} z_j^2e^2  u^{\mathrm{P}}(\mathbf{r};1),
\label{PCA_mu}
\end{equation}
which ends the proof of this Lemma.
\end{proof}

Lemma \ref{lemmaPCA} implies that the chemical potential is exactly the self energy at the full charging $\lambda=1$, and thus the Debye charging process is consistent with G\"untelberg charging process in the PCA situation. With this lemma, we can easily obtain the simplified expression for the correlation energy in the higher order approximations.

\begin{lemma}\label{lemmaHO}
Let \begin{equation}
\displaystyle \mathcal{F}_{\mathrm{cor}}^{(l)}= \int_\Omega \sum_i c_i(\mathbf{p}) z_i^2 e^2 \int_0^1 \lambda u^{(l)}(\mathbf{p};\lambda) d\lambda d\mathbf{p}
\label{First_F}
\end{equation}
be the correlation component of the free energy with $l=0$, 1 or 2 to represent the order of asymptotics. We have:

(1) The zeroth-order and first-order excess chemical potentials are,
 \begin{equation}
 \mu_{j,\mathrm{cor}}^{(l)}=  \delta \mathcal{F}_{\mathrm{cor}}^{(l)}/\delta c_j=   z_j^2e^2  u^{(l)}(\mathbf{r};1)/2,~~\hbox{for} ~l=0, 1. \label{order01}
 \end{equation}

(2) The second-order asymptotics does not have the same conclusion, but it holds the reciprocal relation,
\begin{equation}
\frac{z_i^2 e^2}{2}\frac{\delta u_i^{(2)}(\mathbf{r};\lambda)}{\delta c_j(\mathbf{r}')}=\frac{z_j^2 e^2}{2}\frac{\delta u_j^{(2)}(\mathbf{r}';\lambda)}{\delta c_i(\mathbf{r})},
\label{Second_rep}
\end{equation}
which implies there exists a free energy functional whose variation is the self energy
at the full charging $ \delta \mathcal{F^*}/\delta c_j = z_j^2e^2 u^{(2)}/2$. We will redefine $\mathcal{F}_{\mathrm{cor}}^{(2)}=\mathcal{F}^*$.
\end{lemma}

\begin{proof}
(1)  Since the Born energy is independent of the ionic concentration, we have,
\begin{equation}
\displaystyle \frac{\delta }{\delta c_j(\mathbf{r})} \int_\Omega \int_0^1 \lambda \sum_i c_i(\mathbf{p}) z_i^2 e^2 \left[\frac{1}{\varepsilon(\mathbf{r})}-\frac{1}{\varepsilon_0}\right] \frac{1}{4\pi a} d\lambda d\mathbf{p}=\left[\frac{1}{\varepsilon(\mathbf{r})}-\frac{1}{\varepsilon_0}\right] \frac{z_j^2 e^2}{8\pi a}.
\end{equation}
Noticing that $\displaystyle \kappa^2 = \frac{\beta e^2}{\varepsilon}\sum c_i z_i^2$, we can obtain,
\begin{equation}
\displaystyle \frac{\delta }{\delta c_j(\mathbf{r})} \int_\Omega \int_0^1 \lambda a \frac{\sum_i c_i(\mathbf{p}) z_i^2 e^2 }{4\pi\varepsilon(\mathbf{r})}\left[\lambda^2\kappa^2(\mathbf{p})+\frac{\nabla^2 \varepsilon(\mathbf{p})}{6\varepsilon(\mathbf{p})}\right] d\lambda d\mathbf{p}=\frac{z_j^2 e^2}{2}\Delta u^{(1)}(\mathbf{r};1).
\end{equation}
Using Lemma \ref{lemmaPCA} for the PCA term, we find that
the excess chemical potential equals the self energy of an ion up to the first-order asymptotic expansion. We complete the proof of Eq. \eqref{order01}.

(2) For the second-order approximation, following the same procedure as \eqref{Varfl}, we can easily see that the increment in chemical potential
does not equal the increment in the self energy at full charing,
\begin{equation}
\Delta \mu^{(2)}(\mathbf{r}) \neq \frac{z_j^2e^2}{2}\Delta u^{(2)}(\mathbf{r};1).
\end{equation}
The PCA energy satisfies the reciprocal relation. From the proof of Lemma (2.1), we can straightforwardly obtain that,
\begin{equation}
\frac{1}{2} z_i^2e^2 \frac{\delta \Delta u^{(1)}(\mathbf{r};\lambda)}{\delta c_j(\mathbf{r}')}=\frac{\beta e^4 z_i^2 z_j^2 }{8\pi\varepsilon^2(\mathbf{r})}\lambda^2 a \delta(\mathbf{r}-\mathbf{r'})=\frac{1}{2} z_j^2e^2\frac{\delta \Delta u^{(1)}(\mathbf{r}';\lambda)}{\delta c_i(\mathbf{r})}.
\end{equation}
And now for the second-order correction,
\begin{eqnarray}
&&\frac{1}{2} z_i^2e^2 \frac{\delta \Delta u^{(2)}(\mathbf{r};\lambda)}{\delta c_j(\mathbf{r}')}\nonumber\\
&=&\frac{\beta e^4 z_i^2 z_j^2}{2\varepsilon}\lambda^2 a^2 u^P(\mathbf{r};\lambda)\delta(\mathbf{r}-\mathbf{r'})+\lambda^2 \kappa^2 a^2\frac{1}{2} z_i^2e^2 \frac{\delta u^{P}(\mathbf{r};\lambda)}{\delta c_j(\mathbf{r}')} \nonumber\\
&=&\frac{\beta e^4 z_i^2 z_j^2}{2\varepsilon}\lambda^2 a^2 u^P(\mathbf{r};\lambda)\delta(\mathbf{r}-\mathbf{r'})+\lambda^2 \kappa^2 a^2\frac{1}{2} z_j^2e^2 \frac{\delta u^{P}(\mathbf{r};\lambda)}{\delta c_i(\mathbf{r}')}\nonumber\\
&=& \frac{1}{2} z_j^2e^2\frac{\delta \Delta u^{(2)}(\mathbf{r}';\lambda)}{\delta c_i(\mathbf{r})}.
\end{eqnarray}
This yields the relation \eqref{Second_rep}.

We define a functional through the series expansion around some reference density distribution $\mathbf{c}^0=\{c_i^0(\mathbf{r}), i=1,\cdots, N\}$ as follows,
\begin{eqnarray}
&&\mathcal{F}^*[\mathbf{c}] = \sum_i \int_\Omega  [c_i(\mathbf{r})-c_i^0(\mathbf{r})]U_i^{(2*)}(\mathbf{r};1) d\mathbf{r}  \nonumber \\ &&~~~~~~~+\frac{1}{2}\sum_{i,j} \iint_\Omega [c_i(\mathbf{r})-c_i^0(\mathbf{r})][c_j(\mathbf{r'})-c_j^0(\mathbf{r'})] \frac{\delta U_i^{(2*)}(\mathbf{r};1) }{\delta c_j^0(\mathbf{r})}d\mathbf{r}d\mathbf{r'} + \cdots,
\label{Fstar2}
\end{eqnarray}
where
\begin{equation}
\displaystyle U_i^{(2*)}=\frac{z_i^2 e^2}{2}\left(u^{\mathrm{P*}}+ \left[\frac{1}{\varepsilon(\mathbf{r})}-\frac{1}{\varepsilon_0}\right] \frac{1}{4\pi a} +\Delta u^{(1*)}+\Delta u^{(2*)}\right)
\end{equation}
is the second-order approximated self energy of the $i$th ion in the system with given distributions $\mathbf{c}^0$, and $u^{\mathrm{P*}}$, $\Delta u^{(1*)}$, $\Delta u^{(2*)}$ are correlation energies by
the PCA, the first- and second-order asymptotics (given by Eqs. \eqref{selfenergy_PCA}, \eqref{first_add}, and \eqref{secondorder},  respectively.) with the $\mathbf{c}^0$ as the ion distributions.
By the expression of  $\delta u^{P*}(\mathbf{r};1) / \delta c_j^0(\mathbf{r})$ in Lemma \ref{lemmaCond}, we can estimate the upper bound of each integral in Eq. \eqref{Fstar2},
\begin{eqnarray}
&&\left|\iint [c_{i_0}(\mathbf{r_0})-c_{i_0}^0(\mathbf{r_0})]\cdots [c_{i_n}(\mathbf{r_n})-c_{i_n}^0(\mathbf{r_n})] \frac{\delta^{(n)} U_{i_0}^{(2*)}(\mathbf{r_0};1) }{\delta c_{i_1}^0(\mathbf{r_1})\cdots \delta c_{i_n}^0(\mathbf{r_n})}d\mathbf{r_1}\cdots d\mathbf{r_n}\right|\nonumber \\
&& \leq M \prod_k || z_{i_k}(c_{i_k}-c_{i_k}^0)||_{L^2(\Omega)},
\end{eqnarray}
with some constant $M$. Thus,the series in \eqref{Fstar2} converges and  is well defined.

By the reciprocal relation,
\begin{equation}
\displaystyle \frac{\delta U_i^{(2*)}(\mathbf{r};1) }{\delta c_j^0(\mathbf{r})}=\frac{\delta U_j^{(2*)}(\mathbf{r};1) }{\delta c_i^0(\mathbf{r})},
\end{equation}
the corresponding excess chemical potential can be written as,
\begin{eqnarray}
\mu_i^*(\mathbf{r})&\triangleq&\frac{\delta \mathcal{F}^*}{\delta c_i(\mathbf{r})}
=U_i^{(2*)}(\mathbf{r};1)+\sum_j \int [c_j(\mathbf{r'})-c_j^0(\mathbf{r'})] \frac{\delta U_i^{(2*)}(\mathbf{r};1) }{\delta c_j^0(\mathbf{r})}d\mathbf{r'}+\cdots\nonumber\\
&=&U_i^{(2)}(\mathbf{r};1).
\end{eqnarray}
The proof is then completed.
\end{proof}

From Lemma \ref{lemmaHO}, we can see when the ion size is introduced, the Debye charging process and G\"untelberg charging process are consistent to the first-order asymptotic approximation,
while they are not consistent for the second-order approximation.

Now we obtained the total free energy with approximate correlation energy,
\begin{equation}
F=\int_\Omega \left[ \frac{\varepsilon}{2} |\nabla \phi |^2+\sum_{i=1}^N k_B T c_i (\log c_i-1)  \right] d\mathbf{r} + \mathcal{F}_{\mathrm{cor}}^{(l)},
\end{equation}
with $l=0, 1$ and 2, and approximate total chemical potential $\mu_i^{(l)}=z_ie\phi+k_BT \log c_i + \mu_{j,\mathrm{cor}}^{(l)}.$ Here the case of $l=2$ refers to the the functional $\mathcal{F}^*$.
We choose the entropy production,
\begin{equation}
\Delta =\sum_i   \int_\Omega  \frac{D_i}{c_i}  |J_i|^2 d\mathbf{r},
\end{equation}
where $J_i$ is the flux of $i$th ion. The variational approach\cite{Xu:2014:CMS} gives the flux $J_i=D_i c_i\nabla \mu_i$, which is equivalent to the Fick's law.  By coupling the Poisson equation for the electric potential, the GDH equation for the PCA self energy with the flux equations, we get the modified PNP equations
for the charged system, expressed as,
\begin{equation} \label{pnp}
\begin{cases}
\displaystyle \frac{\partial}{\partial t} c_i(\mathbf{r}) = \nabla \cdot D_i \left[ \nabla c_i + \beta \left( z_i e c_i \nabla \phi +  \frac{1}{2} z_i^2e^2 c_i  \nabla u^{(l)}\right) \right],\\
\displaystyle -\nabla \cdot \varepsilon(\mathbf{r}) \nabla \phi = \sum_i z_i e c_i,
\end{cases}
\end{equation}
where $u^{(l)}$ is given by,
\begin{equation}
\begin{cases}
\displaystyle u^{(0)}(\mathbf{r}; 1)= u^{\mathrm{P}}(\mathbf{r};1)+ \left[\frac{1}{\varepsilon(\mathbf{r})}-\frac{1}{\varepsilon_0}\right] \frac{1}{4\pi a}, \\
\displaystyle u^{(1)}(\mathbf{r}; 1)= u^{(0)}(\mathbf{r};1) + \frac{a}{4\pi\varepsilon(\mathbf{r})}\left[ \kappa^2(\mathbf{r})+\frac{\nabla^2 \varepsilon(\mathbf{r})}{6\varepsilon(\mathbf{r})}\right], \\
\displaystyle u^{(2)}(\mathbf{r}; 1)= u^{(1)}(\mathbf{r};1) + \kappa^2(\mathbf{r}) a^2 u^\mathrm{P}(\mathbf{r}; 1),
\end{cases}
\end{equation}
and the PCA self energy is determined by,
\begin{equation}
\begin{cases}
-\nabla \cdot \varepsilon(\mathbf{r})\nabla G^\mathrm{P}(\mathbf{r},\mathbf{r'}) + 2 I(\mathbf{r})  G^\mathrm{P}(\mathbf{r},\mathbf{r'}) =\delta(\mathbf{r}-\mathbf{r'}),\\
\displaystyle u^{\mathrm{P}}(\mathbf{r};1)=\lim_{\mathbf{r}'\to\mathbf{r}} \left[ G^\mathrm{P}(\mathbf{r};\mathbf{r'}) -\frac{1}{\varepsilon(\mathbf{r})}G_0(\mathbf{r};\mathbf{r'})\right].
\end{cases} \label{PCA_energy}
\end{equation}
We have the following theorem for the modified PNP system.

\begin{theorem} The set of modified PNP equations \eqref{pnp} for $l=0, 1$ or 2 describes a dissipation system with a proper energy dissipation law.
\end{theorem}

\begin{proof} Let $\Gamma=\partial \Omega$ be the boundary of the electrolyte. By Lemma \ref{lemmaHO},
the time evolution of the energy is,
\begin{eqnarray}
&\displaystyle \frac{d F}{dt}& = \frac{d}{dt}\int_\Omega \left[\frac{\varepsilon|\nabla\phi|^2}{2}+k_BT \sum_i c_i(\log c_i-1) \right]d\mathbf{r} + \frac{d}{dt} \mathcal{F}_{\mathrm{cor}}^{(l)} \nonumber \\
&&= \int_\Omega \left[\varepsilon \nabla\phi\cdot\nabla\phi_t+ \sum_i \left(k_BT\log c_i \frac{\partial}{\partial t}c_i + \frac{1}{2}z_i^2e^2 u^{(l)}  \frac{\partial}{\partial t}c_i \right) \right]d\mathbf{r}.
\end{eqnarray}
Integration by parts and using the Poisson equation gives,
\begin{eqnarray}
&\displaystyle \frac{d F}{dt}&= \int_\Omega \left[ -\phi \nabla \cdot \varepsilon \nabla \phi_t +\sum_i \left(k_BT\log c_i + \frac 12 z_i^2e^2 u^{(l)} \right)\frac{\partial c_i}{\partial t} \right]d\mathbf{r}+ \int_\Gamma \left(\varepsilon \frac{\partial \phi}{\partial \vec{n}}\right)_t \phi d\mathbf{p} \nonumber \\
&&= \int_\Omega \left[ \sum_i \left(z_i e\phi+ k_BT\log c_i + \frac 12 z_i^2e^2 u^{(l)} \right)\frac{\partial c_i}{\partial t}  \right]d\mathbf{r}+ \int_\Gamma \left(\varepsilon \frac{\partial \phi}{\partial \vec{n}}\right)_t \phi d\mathbf{s} \nonumber\\
&&= \int_\Omega \left[ \sum_i \mu_i^{(l)} \nabla \cdot D_i c_i \nabla \mu_i^{(l)} \right]d\mathbf{r}+ \int_\Gamma \left(\varepsilon \frac{\partial \phi}{\partial \vec{n}}\right)_t \phi d\mathbf{s} \nonumber \\
&&= -\int_\Omega  \sum_i  D_i c_i |\nabla \mu_i^{(l)}|^2 d\mathbf{r}+ \int_\Gamma \left(\varepsilon \frac{\partial \phi}{\partial \vec{n}}\right)_t \phi d\mathbf{s} + \int_\Gamma \sum_i \mu_i D_i c_i \frac{\partial \mu_i^{(l)}}{\partial \vec{n}}  d\mathbf{s}. \label{dissipation}
\end{eqnarray}
The third equality uses the definition of the chemical potential and the Nernst-Planck equations, and the fourth equality uses the integration by parts again.
We can see that the first term in Eq. \eqref{dissipation} is the energy dissipation in domain $\Omega$. The second term stands for the electrostatic energy
input from the boundary by noting that $\left(\varepsilon \partial \phi/\partial \vec{n}\right)_t$ is the current flow to the surface.
In the third term of Eq. \eqref{dissipation}, $D_i c_i \partial \mu_i^{(l)}/\partial \vec{n}$ is the concentration flux through $\Gamma$ and
this term presents the energy transfer through the boundary. Without energy input from the boundary, the system exactly satisfies the
energy dissipation principle.
\end{proof}

\section{Numerical solution for the modified PNP equations}

We show the numerical calculation for the modified PNP equations in this section, in particular, the effect of correlation term.
we denote $U_i(\mathbf{c})=z_i^2e^2 u^{(l)}/2$ and $W_i=z_ie\phi+U_i$ and we introduce the Slotboom variables \cite{Slotboom1969}
$g_i=c_i e^{W_i}.$
We consider a geometry with homogeneity in $y-z$ plane, and the interval for $x$-coordinate is between $[-L, L]$ with Dirichlet boundary conditions for the boundary potential
$\phi(\pm L,t)=V_\pm$ and non-flux boundary conditions $\partial_x g_i=0$ at $x=\pm L$. We assume $D_i=D$ and $\varepsilon$ are
constant.

The generalized Deby-H\"uckel equation for the PCA self-energy is solved through select inversion algorithm \cite{LYM+:ATMS:2011} which could compute the diagonal entries efficiently without computing the whole inverse matrix. This procedure is described clearly in literature \cite{XuMaggs:JCP:14}.
We focus on the discretization of the modified PNP equations and present a numerical scheme which will be proven to remain the mass
conservation of the system due to the non-flux boundary conditions. This semi-implicit scheme is not well studied in the PNP literature,
and we then describe it in some details. Since $c_i=g_i e^{-W_i}$, the PNP equations can be rewritten as,
\begin{eqnarray}
&&\varepsilon\phi_{xx}=- \sum_i z_ie c_i,\quad x\in[a,b],\label{poission} \\
&&\partial_t c_i =D\nabla\cdot\left(e^{-W_i}\nabla g_i\right),\quad x\in[a,b],t>0,
\end{eqnarray}
where $W_i$ and $g_i$ are functions of $c_i$ and $\phi$.

Let $h$ and $k$ be the space and time mesh sizes. We partition $[-L,L]$ with the interior grid points $x^j=-L+h(j-\frac{1}{2}), j=1,...,J$ and
boundary points $x^{\frac{1}{2}}=-L$ and $x^{J+\frac{1}{2}}=+L$. Thus the Poisson equation (\ref{poission}) can be discretized as
\begin{equation}\label{poisson}
\varepsilon(\phi^{j+1}-2\phi^j+\phi^{j-1})  =- h^2 \sum_i z_i e c_i^j,
\end{equation}
with boundary conditions $\frac{1}{2}(\phi^1+\phi^0)=V_-$ and $\frac{1}{2}(\phi^{J+1}+\phi^J)=V_+.$
The discretization is of the second order of accuracy in space. The finite-difference discretization for the evolution of the Nernst-Planck equation
is written as,
\begin{equation}\label{schemetimespace}
\frac{c^{j,n+1}_i-c^{j,n}_i}{k}=\frac{D}{h}\left(e^{-W^{j+1/2,n+1/2}_i}\widehat{g_x}_i^{j+1/2}-e^{-W^{j-1/2,n+1/2}_i}\widehat{g_x}_i^{j-1/2}\right),
\end{equation}
where
\begin{eqnarray}
&&\widehat{g_x}_i^{j+1/2}=\frac{(c^{j+1,n+1}_i+c^{j+1,n}_i)e^{W^{j+1,n+1/2}_i}-(c^{j,n+1}_i+c^{j,n}_i)e^{W^{j,n+1/2}_i}}{2h},\\
&&W^{j,n+1/2}_i=\frac{3W^{j,n}_i-W^{j,n-1}_i}{2},\label{schemetimespacebc}
\end{eqnarray}
with boundary conditions $\widehat{g_x}_i^{1/2}=\widehat{g_x}_i^{N+1/2}=0.$ This is a second-order scheme, and since the use of the
extrapolation for the nonlinear term the iteration for numerical algebra is avoided. It can be shown that the scheme (\ref{schemetimespace})
satisfies the mass conservation, i.e.,
\begin{eqnarray}
\sum_j c_i^{j,n+1}h=\sum_j c_i^{j,n}h.
\end{eqnarray}
This can be easily validated by summing Eq. \eqref{schemetimespace} over index $j$ and using boundary conditions \eqref{schemetimespacebc}.
The scheme is stable under the condition $k=O(h)$ due to the implicit treatment of the second-order differentiation.

\begin{figure*}[htbp]
\begin{center}
\includegraphics[width=0.495\textwidth]{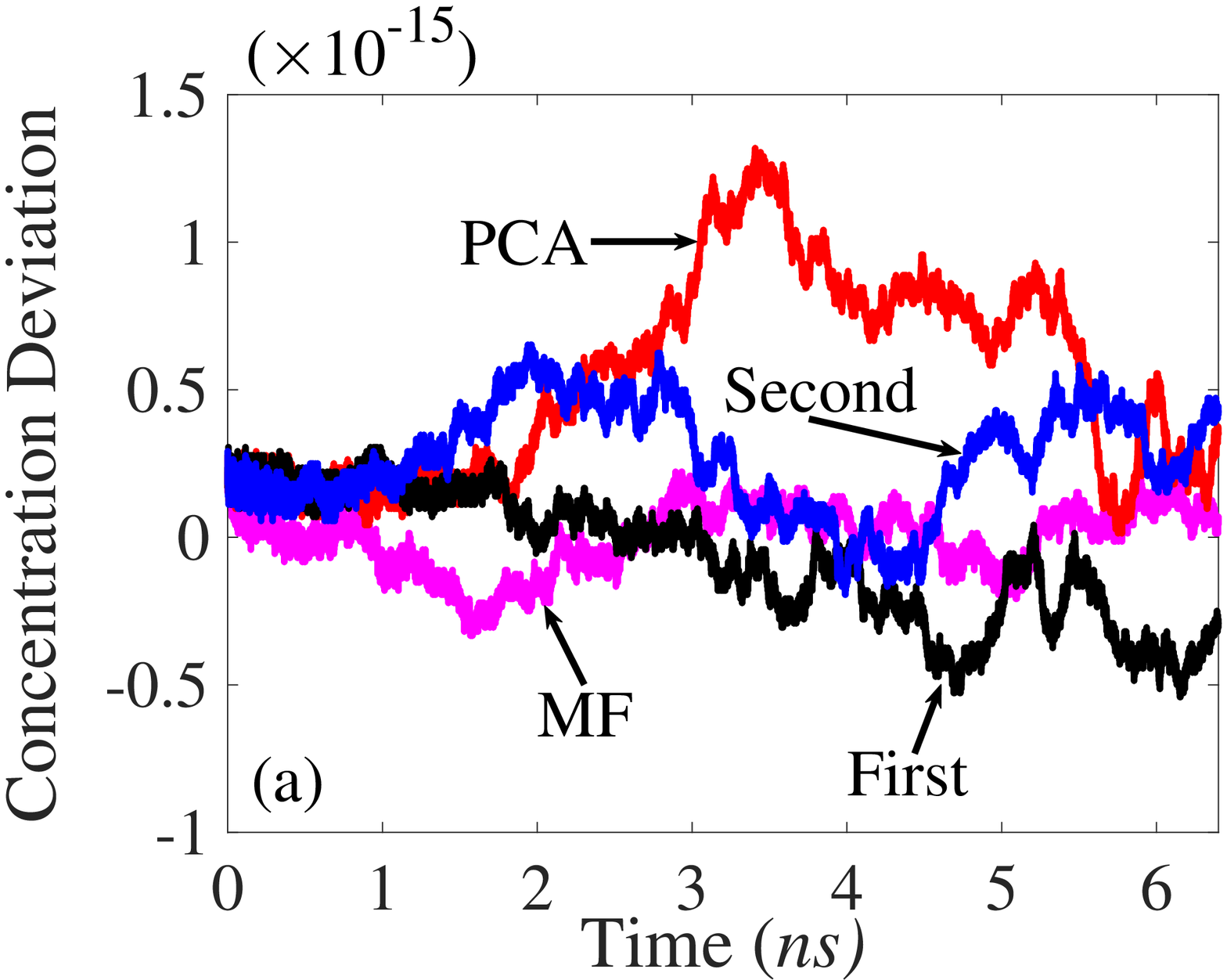}
\includegraphics[width=0.495\textwidth]{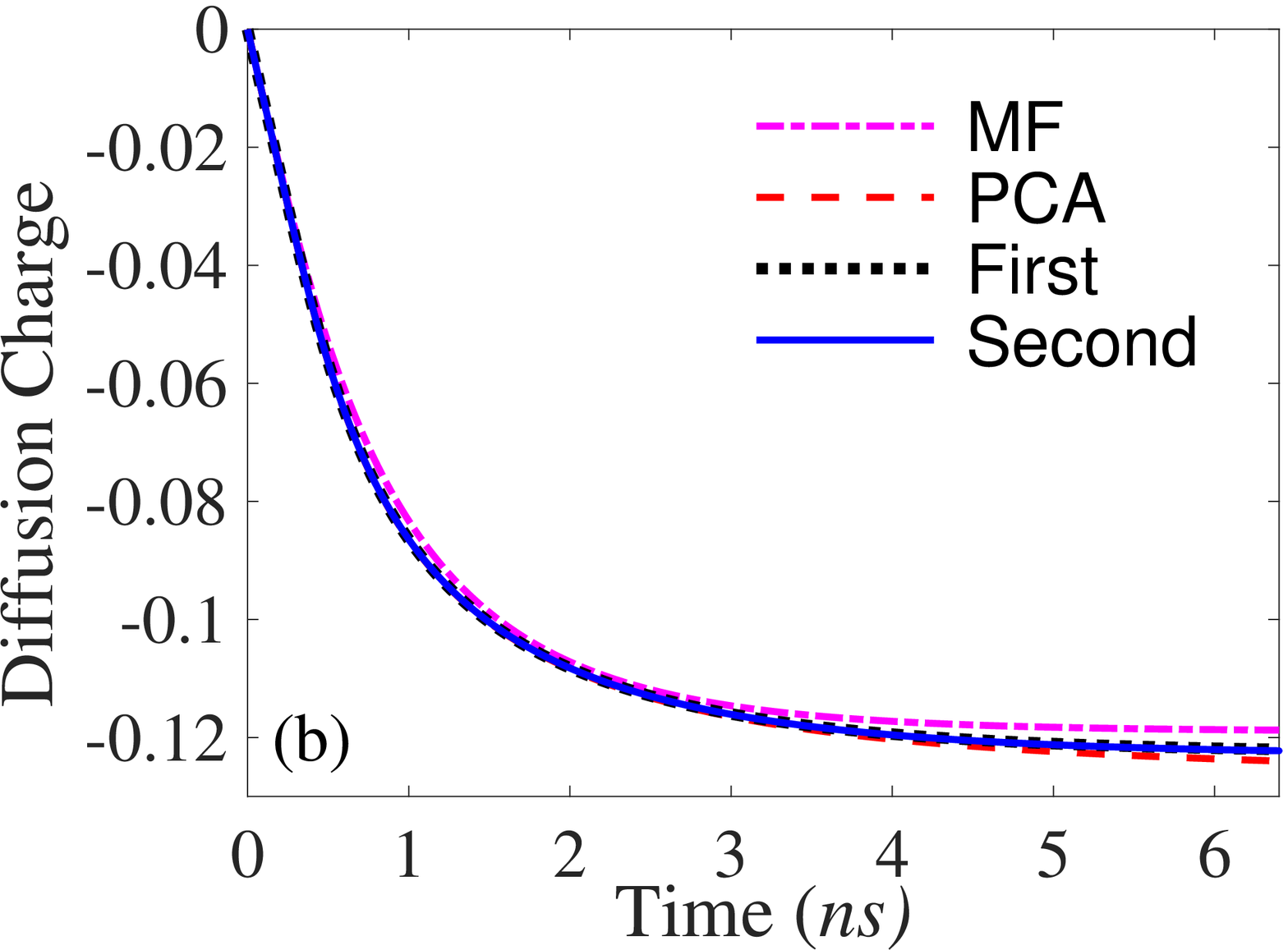}\\
\includegraphics[width=0.495\textwidth]{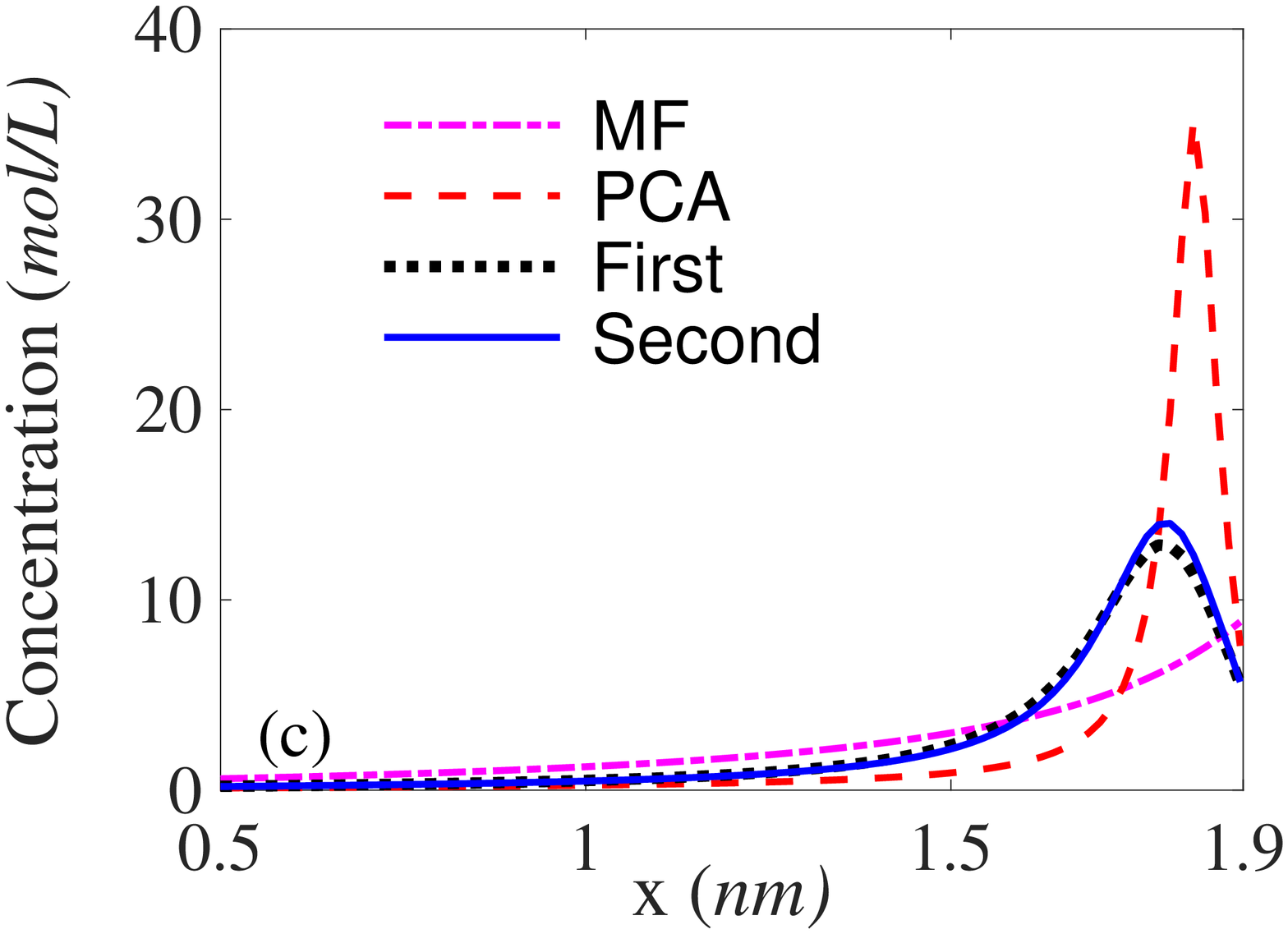}
\includegraphics[width=0.495\textwidth]{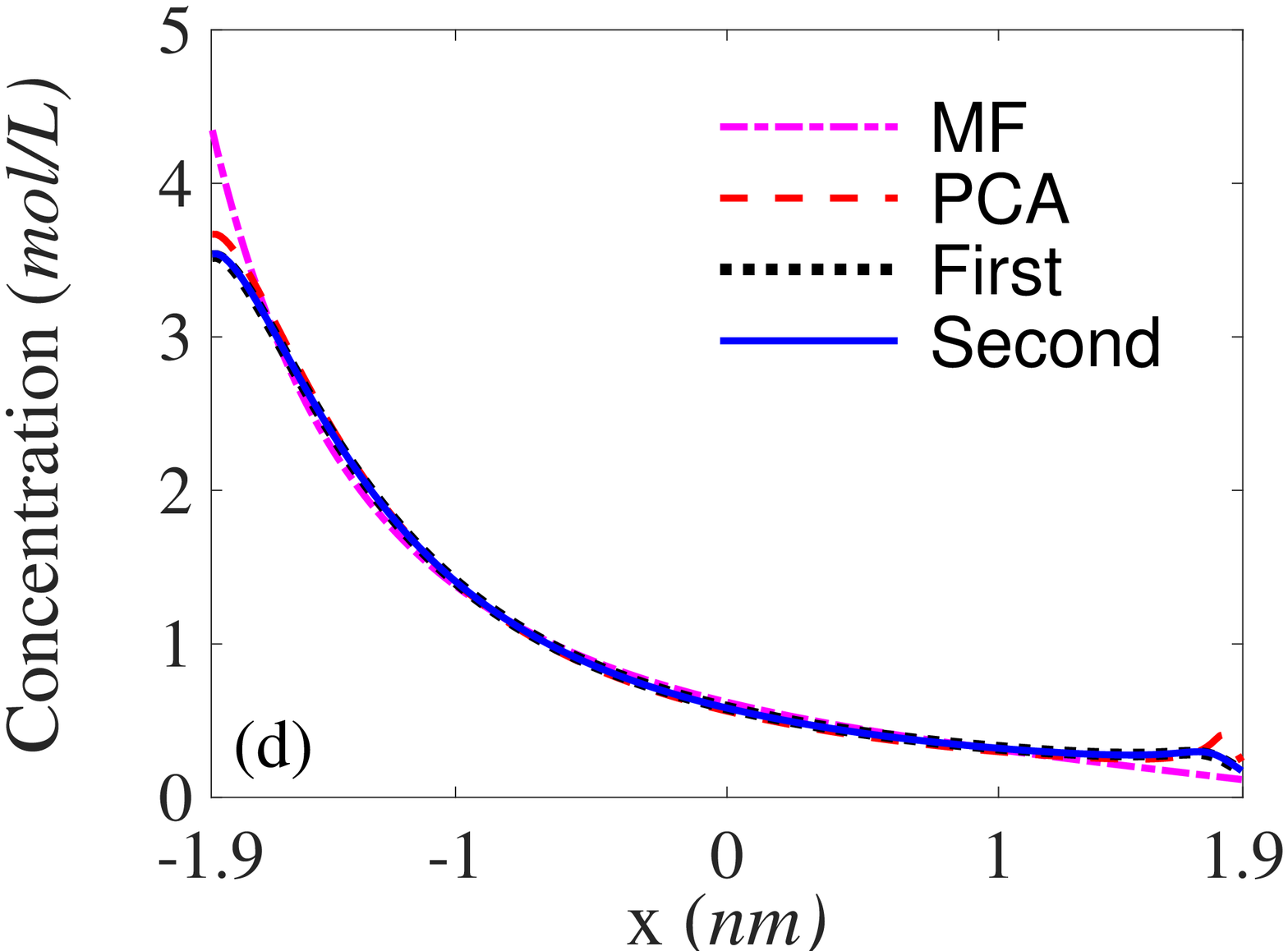}
\caption{Total cation mass, diffusion charge, and equilibrium-state distribution predicted from traditional and modified PNP equations. (a) Deviation of the total cation mass from the exact solution as function of time $t$. (b) Total diffused charge as function of time $t$. (c) The cation distributions at $t=6.4ns$ (the data for $x<0.5$ are not shown as the ion concentration of divalent cations is almost zero near the left electrode). (d) The anion distribution at $t=6.4ns$. }
\label{acc}
\end{center}
\end{figure*}

In our example, we consider an electrolyte with 2:-1 salt between two electrodes with separation of $2L=4 nm$. We take room temperature $T=300K$, the dielectric constant of water $\varepsilon=80\varepsilon_0$ and that of electrodes $\varepsilon=2\varepsilon_0$, ion radius $a=0.1nm$, diffusion coefficient $D=1nm^2/ns$ and the boundary voltages $\phi(\pm L)=\mp 50mV$. The initial values for the anion density is $c_-(x, t=0)=66mM$, and for the cation density $c_+(x, t=0)=33mM$ over the whole domain. The computational domain for the generalized Debye-H\"uckel equation is from $-4nm$ to $4nm$ with periodic conditions on the boundary. The space and time step sizes are $h=0.016 nm$ and $k=0.00032 ns$. We calculate the results up to time $T=6.4 ns$.

In Figure \ref{acc} (a), we calculate the total cation mass, $\sum_j c_+^j h$,  between two electrodes, and present the deviation from the exact value as function of time $t$. The results for anion distribution are similar and thus not shown. The mean-field PNP and the modified PNP with different approximations for correlation energy are all ploted. It can be clearly observed that all the curves only slightly deviate from zero under the machine error, and the numerical scheme preserves the mass conservation, demonstrating the attractive performance of the mass conservation scheme. Figure \ref{acc} (b) displays the the total diffusion charge, $\rho(t)=\int_{-L}^0 \sum_i z_i c_i(x,t) dx$, as function of time $t$, which describes the charges moving from the right half to the left due to the voltage bias on the boundaries. It can be observed that the dynamics of all the four models are quite similar when $t$ is small and the three self modified PNPs predict more diffusion charge which mobiles to the opposite half with a slower dynamics. And the PCA has the largest amount of diffusion charge, slightly bigger than those from the first- and second-order asymptotics. It can be noted that the PCA continues to increase the diffusion charge after the other three models reach equilibrium, indicating the possibility of unstable solution of the PCA. The instability of the PCA can be clearly shown in their ionic distributions at time $T$. The cation and anion concentrations are presented in Figure \ref{acc} (c) and (d), where the slight difference in the first- and second-order asymptotics can be observed in the divalent ion distribution. The mean-field PNP shows monotone ionic distribution and underestimates the counterion condensation due to the neglect of ionic correlation \cite{GNS:RMP:2002,Levin:RPP:2002}. Oppositely, the use of the PCA strongly overestimates the correlation energy. This ascribes to the lack of excluded-volume effect in the PCA, leading to unphysical enhancement of the Coulomb correlation. The depletion zone near the interface is due to the dielectric-boundary effect \cite{XML:PRE:2014} which has been ignored by the mean-field PNP equations. If we keep the calculation up to a longer time, the modified PNP with the PCA would blow up, which is consistent with Ref. \cite{XuMaggs:JCP:14}. These results illustrate that it is essential to use the model with the ion-size effect to remedy this instability for systems with multivalent ions.

\section{Conclusion}

In conclusion, we have developed modified PNP model with contribution from Coulomb correlation where
the excluded volume of ions is taken into account. We start from the free energy functional where
asymptotic approximations for the correlation energy with the ionic size as a small parameter are discussed.
The modified PNP equations are then resulted by the variation with the approximate energy functional.
By numerical example, we demonstrate the new model is useful in predicting the ion structure and dynamics
near interfaces when the Coulomb correlation plays role in the application systems.

\section*{Acknowledgement}
P. L. and Z. X. acknowledge the support from the Natural Science Foundation of China (Grant Nos: 11101276 and 91130012) and the Central Organization Department of China, and the HPC Center of SJTU. X. J. acknowledges the support from the Natural Science Foundation of China (Grant Nos. 11271018 and 91230203) and the National Center for Mathematics and Interdisciplinary Science at CAS. Part of this work was completed during the visits of P. L. and Z. X. to the Beijing Computational Science Research Center in the summer of 2016. The comments from Professor Wei Cai and the hospitality of the institution during the visits are greatly appreciated.


\begin{thebibliography}{10}

\bibitem{FPP+:RMP:2010}
\textsc{R.~H. French, V.~A. Parsegian, R.~Podgornik, R.~F. Rajter, A.~Jagota,
  J.~Luo, D.~Asthagiri, M.~K. Chaudhury, Y.-M. Chiang, S.~Granick, S.~Kalinin,
  M.~Kardar, R.~Kjellander, D.~C. Langreth, J.~Lewis, S.~Lustig, D.~Wesolowski,
  J.~S. Wettlaufer, W.-Y. Ching, M.~Finnis, F.~Houlihan, O.~A. von Lilienfeld,
  C.~J. van Oss, and T.~Zemb}, \emph{Long range interactions in nanoscale
  science}, Rev. Mod. Phys., 82 (2010), pp.~1887--1944.

\bibitem{Schoch:RMP:08}
\textsc{R.~B. Schoch, J.~Han, and P.~Renaud}, \emph{Transport phenomena in
  nanofluidics}, Rev. Mod. Phys., 80 (2008), pp.~839--883.

\bibitem{BTA:PRE:04}
\textsc{M.~Z. Bazant, K.~Thornton, and A.~Ajdari}, \emph{Diffuse-charge
  dynamics in electrochemical systems}, Phys. Rev. E, 70 (2004), p.~021506.

\bibitem{MRS:S:90}
\textsc{P.~Markowich, C.~Ringhofer, and C.~Schimeiser}, \emph{Semiconductor},
  Springer, 1990.

\bibitem{Eisenberg:ACPip:2011a}
\textsc{B.~Eisenberg}, \emph{Crowded charges in ion channels}, Adv. Chem.
  Phys., 148 (2011), pp.~77--223.

\bibitem{MRA:BJ:03}
\textsc{A.~B. Mamonov, R.~D. Coalson, A.~Nitzan, and M.~G. Kurnikova},
  \emph{The role of the dielectric barrier in narrow biological channels: A
  novel composite approach to modeling single-channel currents}, Biophys. J.,
  84 (2003), pp.~3646 -- 3661.

\bibitem{CKC:BJ:03}
\textsc{B.~Corry, S.~Kuyucak, and S.-H. Chung}, \emph{Dielectric self-energy in
  {Poisson-Boltzmann} and {Poisson-Nernst-Planck} models of ion channels},
  Biophys. J., 84 (2003), pp.~3594--3606.

\bibitem{KBA:PRE:2007}
\textsc{M.~S. Kilic, M.~Z. Bazant, and A.~Ajdari}, \emph{Steric effects in the
  dynamics of electrolytes at large applied voltages. {II.} {M}odified
  {P}oisson--{N}ernst--{P}lanck equations}, Phys. Rev. E, 75 (2007), p.~021503.

\bibitem{EHL:JCP:2010}
\textsc{B.~Eisenberg, Y.~Hyon, and C.~Liu}, \emph{Energy variational analysis
  of ions in water and channels: Field theory for primitive models of complex
  ionic fluids}, J. Chem. Phys., 133 (2010), p.~104104.

\bibitem{LZ:BJ:2011}
\textsc{B.~Lu and Y.~Zhou}, \emph{{Poisson-Nernst-Planck} equations for
  simulating biomolecular diffusion-reaction processes {II}: Size effects on
  ionic distributions and diffusion-reaction rates}, Biophys. J., 100 (2011),
  pp.~2475--2485.

\bibitem{JL:2012:JDDE}
\textsc{S.~Ji and W.~Liu}, \emph{{Poisson-Nernst-Planck} systems for ion flow
  with density functional theory for hard-sphere potential: {I-V} relations and
  critical potentials. {Part I: Analysis}}, J. Dyn. Diff. Equ., 24 (2012),
  pp.~955--983.

\bibitem{LTZ:2012:JDDE}
\textsc{W.~Liu, X.~Tu, and M.~Zhang}, \emph{{Poisson-Nernst-Planck} systems for
  ion flow with density functional theory for hard-sphere potential: {I-V}
  relations and critical potentials. {Part II: Numerics}}, J. Dyn. Diff. Equ.,
  24 (2012), pp.~985--1004.

\bibitem{Gillespie:MN:14}
\textsc{D.~Gillespie}, \emph{A review of steric interactions of ions: Why some
  theories succeed and others fail to account for ion size}, Microfluid.
  Nanofluid.,  (2014), pp.~1--22.

\bibitem{Frydel:ACP:2016}
\textsc{D.~Frydel}, \emph{Mean-field electrostatics beyond the point-charge
  description}, Adv. Chem. Phys., 160 (2016), pp.~209--260.

\bibitem{Giera2013}
\textsc{B.~Giera, N.~Henson, E.~M. Kober, T.~M. Squires, and M.~S. Shell},
  \emph{Model-free test of local-density mean-field behavior in electric double
  layers}, Phys. Rev. E, 88 (2013), p.~011301.

\bibitem{Giera2015}
\textsc{B.~Giera, N.~Henson, E.~M. Kober, M.~S. Shell, and T.~M. Squires},
  \emph{Electric double-layer structure in primitive model electrolytes:
  Comparing molecular dynamics with local-density approximations}, Langmuir, 31
  (2015), pp.~3553--3562.

\bibitem{Wu:JCP:2002}
\textsc{Y.-X. Yu and J.~Wu}, \emph{Structures of hard-sphere fluids from a
  modified fundamental-measure theory}, J. Chem. Phys., 117 (2002),
  pp.~10156--10164.

\bibitem{Roth:JPC:2002}
\textsc{R.~Roth, R.~Evans, A.~Lang, and G.~Kahl}, \emph{Fundamental measure
  theory for hard-sphere mixtures revisited: the white bear version}, J. Phys.:
  Condens. Matter, 14 (2002), p.~12063.

\bibitem{BKS+:ACIS:2009}
\textsc{M.~Z. Bazant, M.~S. Kilic, B.~D. Storey, and A.~Ajdari}, \emph{Towards
  an understanding of induced-charge electrokinetics at large applied voltages
  in concentrated solutions}, Adv. Colloid Interface Sci., 152 (2009),
  pp.~48--88.

\bibitem{BSK:PRL:2011}
\textsc{M.~Z. Bazant, B.~D. Storey, and A.~A. Kornyshev}, \emph{Double layer in
  ionic liquids: Overscreening versus crowding}, Phys. Rev. Lett., 106 (2011),
  p.~046102.

\bibitem{bazant2012}
\textsc{B.~D. Storey and M.~Z. Bazant}, \emph{{Effects of electrostatic
  correlations on electrokinetic phenomena}}, Phys. Rev. E, 86 (2012),
  p.~056303.

\bibitem{LE:JPCB:13}
\textsc{J.-L. Liu and B.~Eisenberg}, \emph{Correlated ions in a calcium channel
  model: {A Poisson-Fermi} theory}, J. Phys. Chem. B, 117 (2013),
  pp.~12051--12058.

\bibitem{Liu2016}
\textsc{J.-L. Liu, H.-j. Hsieh, and B.~Eisenberg}, \emph{{Poisson--Fermi
  Modeling of the Ion Exchange Mechanism of the Sodium/Calcium Exchanger}}, J.
  Phys. Chem. B, 120 (2016), pp.~2658--2669.

\bibitem{AM:CJU:1986}
\textsc{S.~M. Avdeev and G.~A. Martynov}, \emph{Influence of image forces on
  the electrostatic component of the disjoining pressure}, Colloid J. USSR, 48
  (1986), pp.~535--542.

\bibitem{FM:PRE:2016}
\textsc{D.~Frydel and M.~Ma}, \emph{Density functional formulation of the
  random-phase approximation for inhomogeneous fluids: {A}pplication to the
  {G}aussian core and {C}oulomb particles}, Phys. Rev. E, 93 (2016), p.~062112.

\bibitem{podgornik1989jcp}
\textsc{R.~Podgornik}, \emph{Electrostatic correlation forces between surfaces
  with surface specific ionic interactions}, J. Chem. Phys., 91 (1989),
  pp.~5840--5849.

\bibitem{NO:EPJE:2000}
\textsc{R.~R. Netz and H.~Orland}, \emph{Beyond {Poisson-Boltzmann}:
  Fluctuation effects and correlation functions}, Eur. Phys. J. E, 1 (2000),
  pp.~203--214.

\bibitem{NO:EPJE:2003}
\textsc{R.~R. Netz and H.~Orland}, \emph{Variational charge renormalization in
  charged systems}, Eur. Phys. J. E, 11 (2003), pp.~301--311.

\bibitem{XuMaggs:JCP:14}
\textsc{Z.~Xu and A.~Maggs}, \emph{Solving fluctuation-enhanced
  {Poisson}-{Boltzmann} equations}, J. Comput. Phys., 275 (2014), pp.~310--322.

\bibitem{Wang:PRE:2010}
\textsc{Z.~G. Wang}, \emph{Fluctuation in electrolyte solutions: The self
  energy}, Phys. Rev. E, 81 (2010), p.~021501.

\bibitem{Bikerman:PM:1942}
\textsc{J.~J. Bikerman}, \emph{Structure and capacity of the electrical double
  layer}, Philos. Mag., 33 (1942), pp.~384--397.

\bibitem{Grahame:CR:1947}
\textsc{D.~C. Grahame}, \emph{The electrical double layer and the theory of
  electrocapillarity}, Chem. Rev., 32 (1947), pp.~441--501.

\bibitem{hatlo2012}
\textsc{M.~M. Hatlo, R.~van Roij, and L.~Lue}, \emph{The electric double layer
  at high surface potentials: The influence of excess ion polarizability},
  Europhys. Lett., 97 (2012), p.~28010.

\bibitem{LWZ:CMS:14}
\textsc{B.~Li, J.~Wen, and S.~Zhou}, \emph{Mean-field theory and computation of
  electrostatics with ionic concentration dependent dielectrics}, Commun. Math.
  Sci., 14(1) (2016), pp.~249--271.

\bibitem{Guan:PRE:16}
\textsc{X.~Guan, M.~Ma, Z.~Gan, Z.~Xu, and B.~Li}, \emph{Hybrid {Monte Carlo}
  and continuum modeling of electrolyte with concentration-induced dielectric
  variations}, Phys. Rev. E, 94 (2016), p.~053312.

\bibitem{Vincze2010}
\textsc{J.~Vincze, M.~Valisk{\'o}, and D.~Boda}, \emph{The nonmonotonic
  concentration dependence of the mean activity coefficient of electrolytes is
  a result of a balance between solvation and ion-ion correlations}, J. Chem.
  Phys., 133 (2010), p.~154507.

\bibitem{liu2015poisson}
\textsc{J.-L. Liu and B.~Eisenberg}, \emph{Poisson--fermi model of single ion
  activities in aqueous solutions}, Chem. Phys. Lett., 637 (2015), pp.~1--6.

\bibitem{DH:PZ:1923b}
\textsc{P.~Debye and E.~H\"uckel}, \emph{The theory of electrolytes. {I.
  Lowering} of freezing point and related phenomena}, Phys. Zeitschr., 24
  (1923), pp.~185--206.

\bibitem{WangRui:JCP:13}
\textsc{R.~Wang and Z.-G. Wang}, \emph{Effects of image charges on double layer
  structure and forces}, J. Chem. Phys., 139 (2013), p.~124702.

\bibitem{MX:JCP:14}
\textsc{M.~Ma and Z.~Xu}, \emph{Self-consistent field model for strong
  electrostatic correlations and inhomogeneous dielectric media}, J. Chem.
  Phys., 141 (2014), p.~244903.

\bibitem{WW:JCP:2015}
\textsc{R.~Wang and Z.-G. Wang}, \emph{On the theoretical description of weakly
  charged surfaces}, J. Chem. Phys., 142 (2015), p.~104705.

\bibitem{LMX:CiCP:17}
\textsc{P.~Liu, M.~Ma, and Z.~Xu}, \emph{Understanding depletion induced
  like-charge attraction from self-consistent field model}, Commun. Comput.
  Phys., 22 (2017), pp.~95--111.

\bibitem{Flavell:JCE:14}
\textsc{A.~Flavell, M.~Machen, B.~Eisenberg, J.~Kabre, C.~Liu, and X.~Li},
  \emph{A conservative finite difference scheme for {Poisson-Nernst-Planck}
  equations}, J. Comput. Electron., 13 (2014), pp.~235--249.

\bibitem{LW:JCP:14}
\textsc{H.~Liu and Z.~Wang}, \emph{A free energy satisfying finite difference
  method for {Poisson-Nernst-Planck} equations}, J. Comput. Phys., 268 (2014),
  pp.~363 -- 376.

\bibitem{Xu:2014:CMS}
\textsc{S.~Xu, P.~Sheng, and C.~Liu}, \emph{An energetic variational approach
  for ion transport}, Commun. Math. Sci., 12 (2014), pp.~779--789.

\bibitem{HM::2006}
\textsc{J.~P. Hansen and I.~R. McDonald}, \emph{Theory of simple liquids},
  Academic Press, Amsterdam, 2006.

\bibitem{loeb1951}
\textsc{A.~L. Loeb}, \emph{An interionic attraction theory applied to the
  diffuse layer around colloid particles. {I}}, J. Colloid Sci., 6 (1951),
  pp.~75--91.

\bibitem{XML:PRE:2014}
\textsc{Z.~Xu, M.~Ma, and P.~Liu}, \emph{Self-energy-modified
  {Poisson-Nernst-Planck} equations: {WKB} approximation and finite-difference
  approaches}, Phys. Rev. E, 90 (2014), p.~013307.

\bibitem{guntelberg1926}
\textsc{E.~G\"untelberg}, \emph{Interaction of ions}, Z. Phys. Chem., 123
  (1926), pp.~199--247.

\bibitem{Bockris1998}
\textsc{J.~O. Bockris and A.~K.~N. Reddy}, \emph{Volume 1: Modern
  Electrochemistry: Ionics}, Springer US, 1998.

\bibitem{LYM+:ATMS:2011}
\textsc{L.~Lin, C.~Yang, J.~C. Meza, J.~Lu, L.~Ying, and W.~E},
  \emph{Sel{I}nv---{An} algorithm for selected inversion of a sparse symmetric
  matrix}, ACM Trans. Math. Softw., 37 (2011), pp.~40:1--40:19.

\bibitem{FXH:JCP:14}
\textsc{F.~Fahrenberger, Z.~Xu, and C.~Holm}, \emph{Simulation of electric
  double layers around charged colloids in aqueous solution of variable
  permittivity}, J. Chem. Phys., 141 (2014), p.~064902.

\bibitem{IBR:CPC:1998}
\textsc{W.~Im, D.~Beglov, and B.~Roux}, \emph{Continuum solvation model:
  Computation of electrostatic forces from numerical solutions to the
  {Poisson-Boltzmann} equation}, Comput. Phys. Commun., 111 (1998), pp.~59--75.

\bibitem{Slotboom1969}
\textsc{J.~W. Slotboom}, \emph{Iterative scheme for 1- and 2-dimensional
  d.c.-transistor simulation}, Electron. Lett., 26 (1969), pp.~677--678.

\bibitem{GNS:RMP:2002}
\textsc{A.~Y. Grosberg, T.~T. Nguyen, and B.~I. Shklovskii}, \emph{The physics
  of charge inversion in chemical and biological systems}, Rev. Mod. Phys., 74
  (2002), pp.~329--345.

\bibitem{Levin:RPP:2002}
\textsc{Y.~Levin}, \emph{Electrostatic corrections: from plasma to biology},
  Rep. Prog. Phys., 65 (2002), pp.~1577--1632.

\end{thebibliography}

\end{document}